\documentclass[aps,prx,twocolumn,showpacs,superscriptaddress]{revtex4-2}

\usepackage[colorlinks,linkcolor=blue,anchorcolor=green,citecolor=blue,urlcolor=cyan]{hyperref}
\usepackage{amsmath,amsfonts,amssymb,mathrsfs}
\usepackage{bm,bbm}
\usepackage{graphicx}
\usepackage{epstopdf}
\usepackage[dvipsnames]{xcolor}
\usepackage{verbatim}
\usepackage[normalem]{ulem}
\usepackage{mathtools}

\usepackage{algorithm}
\usepackage{algpseudocode}

\usepackage{amsthm}
\newtheorem{theorem}{Theorem}
\newtheorem{corollary}{Corollary}[theorem]

\usepackage{physics}

\newcommand{\umpsdI}{\hat{U}_{\rm MPS,I}^{\dagger}}
\newcommand{\umpsdII}{\hat{U}_{\rm MPS,II}^{\dagger}}
\newcommand{\calL}{\mathcal{L}}
\newcommand{\calC}{\mathcal{C}}

\usepackage{tikz}

\begin{document}

\title{Filter-Assisted Quantum Subspace Diagonalization via Wavefunction Sparsity Engineering}
\author{Han Xu}
\affiliation{Computational Materials Science Research Team, RIKEN Center for Computational Science (R-CCS), Hyogo 650-0047, Japan}
\email{han.xu@riken.jp}

\author{Tomonori Shirakawa}
\affiliation{Computational Materials Science Research Team, RIKEN Center for Computational Science (R-CCS), Hyogo 650-0047, Japan}
\affiliation{Quantum Computational Science Research Team, RIKEN Center for Quantum Computing (RQC), Saitama 351-0198, Japan}
\affiliation{Computational Condensed Matter Physics Laboratory, RIKEN Pioneering Research Institute (PRI), Saitama 351-0198, Japan}

\author{Seiji Yunoki}
\affiliation{Computational Materials Science Research Team, RIKEN Center for Computational Science (R-CCS), Hyogo 650-0047, Japan}
\affiliation{Quantum Computational Science Research Team, RIKEN Center for Quantum Computing (RQC), Saitama 351-0198, Japan}
\affiliation{Computational Quantum Matter Research Team, RIKEN Center for Emergent Matter Science (CEMS), Saitama 351-0198, Japan}
\affiliation{Computational Condensed Matter Physics Laboratory, RIKEN Pioneering Research Institute (PRI), Saitama 351-0198, Japan}

\begin{abstract}
    Subspace diagonalization techniques based on quantum sampling, such as quantum selected configuration interaction (QSCI) and sample-based quantum diagonalization (SQD), have recently emerged as promising quantum-centric approaches for approximating ground-state energies of many-body systems.
    However, their performance is fundamentally limited by an intrinsic trade-off between sampling efficiency and the sparsity of the ground-state wavefunction, which becomes particularly severe in strongly correlated systems. 
    Here, we introduce a filter-assisted SQD protocol that engineers wavefunction sparsity via a quantum filter, i.e., a unitary transformation of the Hamiltonian designed to concentrate the ground-state weight onto a small number of computational basis states. 
    Using the Gini coefficient as a robust sparsity measure, we establish a quantitative relationship between wavefunction sparsity and the resource requirements of SQD, providing theoretical bounds on the required subspace dimension and sampling cost.
    To realize the quantum filter, we employ a tensor-network-based circuit-encoding algorithm that maps target states to quantum circuits with controllable fidelity.
    We benchmark our approach on the quantum Ising model with transverse and longitudinal fields using both numerical simulations and quantum hardware experiments.
    Our results demonstrate that, compared with standard SQD, the proposed protocol significantly enhances wavefunction sparsity, reduces ground-state energy estimation errors by orders of magnitude, and substantially lowers sampling overhead. 
    These findings establish filter-assisted subspace diagonalization as a powerful and scalable framework for quantum many-body calculations in the strongly correlated regime.
\end{abstract}

\date{\today}

\maketitle

\section{Introduction}


Accurate determination of ground states and their energies is a central problem in quantum many-body physics, with broad applications in condensed matter physics~\cite{Sachdev2011,Wang2024Certifying} and quantum chemistry~\cite{Lanyon2010,lee2023evaluating}.
Subspace diagonalization provides a general framework for approximating low-lying eigenstates of a large Hermitian operator $\hat H$ by projecting the full Hilbert space onto a reduced subspace spanned by a selected set of basis states. 
Within this subspace, the Hamiltonian (and, if necessary, overlap) matrices are constructed and diagonalized to obtain approximate eigenvalues and eigenvectors~\cite{Huggins2020Non}.
The accuracy and efficiency of subspace diagonalization critically depend on the choice of basis states, making the construction of an effective subspace a central challenge in practice.

In classical computation, the (modified) Lanczos method recursively generates a two-dimensional subspace $\mathrm{span}(\ket{\nu_0}, \hat H\ket{\nu_0})$ to iteratively minimize the energy, or more generally constructs a Krylov subspace $\mathrm{span}(\ket{\nu_0}, \hat H\ket{\nu_0},\dots,\hat H^{k-1}\ket{\nu_0})$ for diagonalization~\cite{lanczos1950iteration,koch2011}.
In quantum computation, subspace diagonalization is naturally suited to quantum-classical hybrid algorithms: a quantum processor can efficiently prepare and manipulate quantum states in the full Hilbert space, 
while a classical processor performs diagonalization within a reduced subspace to numerical precision.
Several such hybrid schemes have been proposed to systematically expand the subspace, using quantum circuits based on Hamiltonian powers~\cite{seki2021power}, real-time evolution~\cite{stair2020multireference,Cortes2022Krylov,parrish2019quantum,cohn2021filter,Bespalova2021}, and imaginary-time evolution~\cite{motta2020determining,yeter2020practical}, as well as eigenvector continuation~\cite{frame2018ec,francis2022subspace} .

More recently, subspace diagonalization schemes in which the subspace is spanned by quantum-selected computational basis states have emerged as promising approaches for near-term quantum computing. 
A representative example is quantum selected configuration interaction (QSCI)~\cite{kanno2023quantum}, and closely related sample-based quantum diagonalization (SQD) methods have subsequently been developed and demonstrated~\cite{robledo2024chemistry_pub}. 
In particular, by introducing a self-consistent configuration recovery technique, Ref.~\cite{robledo2024chemistry_pub} achieved a large-scale demonstration on molecular systems beyond the scale accessible to exact diagonalization. 
Numerous variants of this framework have since been independently proposed~\cite{yu2025quantum,sugisaki2024hamiltonian,mikkelsen2024quantum,nogaki2025symmetry,danilov2025enhancing,barison2025quantum,duriez2025computing,ohgoe2025quantum}. 
Furthermore, SQD methods have recently enabled even larger-scale quantum--HPC hybrid computations~\cite{shirakawa2025a,Merz2016}.

Despite these recent successes, sample-based approaches face a fundamental trade-off between sampling efficiency and the sparsity of the ground-state wavefunction~\cite{reinholdt2025fundamental}. 
If the wavefunction is highly sparse and dominated by a few configurations, repeated sampling tends to return the same dominant configurations, making it difficult to identify new basis states for subspace expansion. 
Conversely, if the wavefunction is less sparse and more broadly distributed, the quantum resources required to reach a given accuracy may grow rapidly, potentially exponentially, with system size. 
Crucially, the ground-state wavefunction itself is not necessarily the optimal distribution for sampling basis states in subspace diagonalization. 
Ideally, an optimal sampler would yield an approximately uniform distribution \emph{only} over the basis states relevant to the target subspace. 
This observation suggests that the trade-off can be mitigated by engineering the sampling distribution via a transformation of the Hamiltonian. 
In particular, one may transform the Hamiltonian such that its ground state becomes sparse in the computational basis, design a sampler tailored to subspace construction, and then apply the standard sample-based protocol to the transformed Hamiltonian.

In this paper, we present a filter-assisted method that engineers the sparsity of the ground-state wavefunction relevant to quantum sampling.
A schematic comparison between standard SQD and the proposed filter-assisted SQD (FSQD) protocol is shown in Fig.~\ref{fig:sqd-alg}. 
Within this framework, we consider two variants, depending on whether a projection is introduced to remove the dominant zero-state component, i.e., $|0\rangle^{\otimes n}$. 
Without this projection, samples from the filtered ground state are strongly concentrated on the zero state, leading to poor sampling efficiency and thereby exposing the intrinsic trade-off of sample-based subspace construction discussed above. 
By contrast, introducing the projection suppresses the zero-state contribution and allows the sampler to explore basis states useful for subspace expansion. 
As a result, the FSQD protocol achieves more accurate ground-state energy estimates with substantially lower measurement overhead than standard SQD. This is demonstrated using classical simulations and quantum-hardware experiments on the quantum Ising model with transverse and longitudinal fields for systems of up to 100 qubits. 

The remainder of this paper is organized as follows. 
In Sec.~\ref{sec:fsqd}, we briefly review the standard SQD method and outline our filter-assisted approach. 
In Sec.~\ref{sec:sparsity}, we introduce the Lorenz curve and Gini coefficient as measures of wavefunction sparsity and establish theoretical bounds relating sparsity to the resource requirements of SQD. 
In Sec.~\ref{sec:automatic}, we present the quantum circuit encoding algorithm used to construct the circuit filter. 
In Sec.~\ref{sec:ising-model}, we benchmark the proposed method against standard SQD for the quantum Ising model with transverse and longitudinal fields, comparing their performance in terms of ground-state energy estimation error and the number of samples required to achieve a given accuracy. 
We further perform an energy-variance analysis and use extrapolation to improve the accuracy of the estimated ground-state energy.
In Sec.~\ref{sec:expr}, we demonstrate both protocols on IBM quantum hardware. 
Finally, we summarize our results in Sec.~\ref{sec:summary}. 
Appendix~\ref{appx:I} proves the monotonic decrease of the ground-state energy estimate with increasing subspace dimension. 
Appendix~\ref{appx:II} derives relations between wavefunction sparsity and the computational resources required in the standard SQD method. 
Appendix~\ref{appx:III} discusses the performance of the MPS-based circuit encoding methods. 
Appendix~\ref{appx:IV} provides additional details on the quantum-hardware experiments.

\section{Filter-Assisted Subspace Diagonalization Method}\label{sec:fsqd}

\begin{figure*}[tb]
    \includegraphics[width=\linewidth]{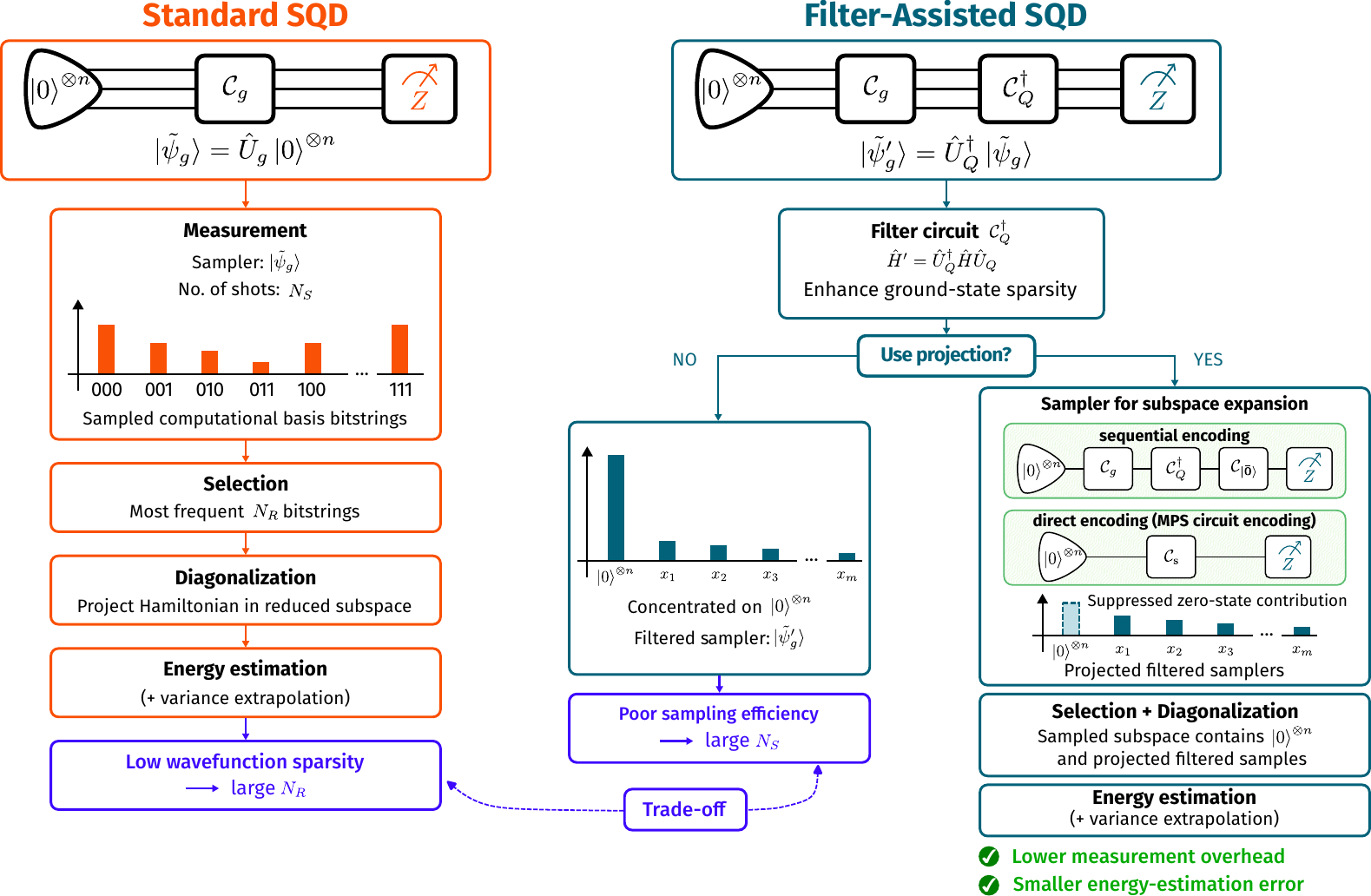}
    \caption{
    Schematic comparison between the standard SQD workflow and the proposed filter-assisted SQD (FSQD) protocol.
    In both methods, a truncated Hamiltonian matrix is constructed and diagonalized in the subspace spanned by the $N_R$ most frequently sampled computational basis states.
    In FSQD, a designed filter circuit $\calC_Q^{\dagger}$ (unitary $\hat U_Q^{\dagger}$) transforms the original Hamiltonian $\hat H$ into the filtered Hamiltonian $\hat H'\coloneqq \hat U_Q^{\dagger}\hat H\hat U_Q$, whose ground state is more strongly concentrated on the computational basis state $\ket{\bm{0}}=|0\rangle^{\otimes n}$. 
    To avoid excessive sampling of the dominant $\ket{\bm 0}$ component, samplers for subspace expansion are constructed from the projected state---either by directly encoding the normalized projected state or by sequentially encoding $\hat{U}_Q^\dagger$ followed by a unitary approximation to the projection step. 
    This suppresses the zero-state contribution and mitigates the trade-off between wavefunction sparsity and sampling efficiency inherent to the standard SQD.
    }
    \label{fig:sqd-alg}
\end{figure*}

In the standard SQD protocol~\cite{kanno2023quantum,robledo2024chemistry_pub}, a quantum state $\ket*{\tilde\psi_g}$ is measured in the computational basis. Here, $\ket*{\tilde\psi_g}$ serves as a sampler and is typically chosen to approximate the target ground state $\ket{\psi_g}$ of Hamiltonian $\hat H$. We assume that $\ket*{\tilde\psi_g}$ is prepared on a quantum computer by applying a quantum circuit $\calC_g$ (with the associated unitary operator $\hat U_g$) to the zero state $\ket{\bm{0}}=|0\rangle^{\otimes n}$ of an $n$-qubit system, i.e., $\ket*{\tilde\psi_g}=\hat U_g\ket{\bm{0}}$. 
Each measurement yields a bitstring $x\in \{0,1\}^{n}$.
After performing $N_S$ measurement shots, one obtains a set of $N_S$ bitstrings, generally with repetitions. 
The frequency of each distinct bitstring is then counted, and the $N_R$ most frequent bitstrings are selected to define a truncated Hilbert subspace. 
Within this sampled subspace, the projected Hamiltonian is constructed and exactly diagonalized to estimate the ground-state energy.

In an idealized setting, one can show that increasing the subspace dimension $N_R$ leads to a monotonically decreasing estimate of the ground-state energy, thereby implying convergence toward the exact ground-state energy. 
The proof is based on extending the sampled subspace from $N_R$ to $N_R+1$ and then analyzing the projected Hamiltonian within an associated two-dimensional Lanczos subspace. 
Full details are given in Appendix~\ref{appx:I}.

In practice, however, the state $\ket*{\tilde\psi_g}$ expressed in the computational basis is often not sufficiently sparse. 
As a result, repeated bitstrings becomes rare even for large $N_S$,  making it difficult to identify the most relevant basis states and degrading the performance of the SQD protocol. 
A quantitative analysis of this effect, together with rigorous bounds connecting SQD accuracy to wavefunction sparsity, is presented in Sec.~\ref{sec:sparsity} and Appendix~\ref{appx:II}.

The central idea of FSQD is to engineer the sampling distribution by applying a unitary transformation to the Hamiltonian. 
Since subspace diagonalization aims to approximate the ground-state energy $\expval*{\hat{H}}{\psi_g}$, it is natural to exploit the  
invariance of this quantity under a similarity transformation:
\begin{equation}\label{eq:sqd-expval}
    \expval*{\hat{H}}{\psi_g} = \mel*{\psi_g}{\hat U_Q\hat U_Q^\dagger\hat{H}\hat U_Q\hat U_Q^\dagger}{\psi_g}.
\end{equation}
Here $\hat{U}_Q$ is the unitary implemented by a designed quantum circuit $\mathcal{C}_Q$, and 
\begin{equation}
\ket{\psi_Q}=\hat U_Q\ket{\bm 0}
\end{equation}
is a classically constructed approximation to the true ground state $|\psi_g\rangle$. 
%

Although $\ket{\psi_Q}$ may only be a coarse approximation, the transformed Hamiltonian
\begin{equation}
\hat{H}' \coloneqq \hat{U}_Q^\dagger \hat{H} \hat{U}_Q
\end{equation}
is expected to have a ground state
\begin{equation}
\ket{\psi_g'} \coloneqq \hat{U}_Q^\dagger \ket{\psi_g}
\end{equation}
and an associated sampler 
\begin{equation}
  \ket{\tilde\psi_g'} \coloneqq \hat{U}_Q^\dagger \ket{\tilde\psi_g} = \hat{U}_Q^\dagger \hat U_g \ket{\bm0},
\end{equation}
both of which are substantially more sparse in the computational basis than $\ket{\psi_g}$ and $\ket{\tilde\psi_g}$, respectively, and in particular more strongly concentrated on $\ket{\bm 0}$.
In this sense, the inverse circuit $\mathcal{C}_Q^\dagger$ acts as a quantum filter on $\hat{H}$.
We then apply the standard SQD procedure to the filtered Hamiltonian $\hat{H}'$ and refer to the resulting protocol as FSQD.

In this work, the circuit $\calC_Q$ is obtained classically by using a tensor-network method to encode an approximate representation of the ground state $\ket{\psi_g}$ into a quantum circuit. 
With this encoding, the zero state $\ket{\bm 0}$ corresponds to the approximate ground state of $\hat H'$, with its energy inherited from the underlying classical tensor-network calculation. 
Consequently, once the sampled subspace includes $\ket{\bm 0}$, FSQD using the sampler $\ket{\tilde\psi_g'}$ can already reproduce this tensor-network estimate and may further improve upon it through subspace expansion. 
A schematic overview of the quantum circuit encoding algorithm is shown in Fig.~\ref{fig:tn-alg}, and the method is discussed in detail in Sec.~\ref{sec:automatic}.

Because the filtered sampler $\ket{\tilde\psi_g'}$ is expected to inherit the strong concentration of $\ket{\psi_g'}$ on $\ket{\bm 0}$, FSQD can in principle achieve high accuracy even when the sampled subspace contains only a small number of basis states. 
However, this same concentration also severely limits subspace expansion through repeated measurements, since the bitstring $\bm 0$ is sampled overwhelmingly often. 
This is precisely the sparsity--sampling trade-off discussed above and highlighted in Ref.~\cite{reinholdt2025fundamental}.

To overcome this difficulty, we separate the role of the dominant $\ket{\bm 0}$ component from that of the remaining basis states useful for subspace expansion. 
Inserting the resolution of identity
\begin{equation}
\hat{I} = \dyad{\bm 0} + \hat{P}_{\ket{\bar{\bm 0}}},
\end{equation}
with
\begin{equation}
\hat{P}_{\ket{\bar{\bm 0}}} \coloneqq \hat{I} - \dyad{\bm 0},
\end{equation}
into Eq.~\eqref{eq:sqd-expval}, we obtain 
\begin{equation}\label{eq:sqd-expval-proj}
    \begin{aligned}
        \expval*{\hat{H}}{\psi_g} =& 
        \mel*{\psi_g}{\hat U_Q\hat U_Q^\dagger\hat{H}\hat U_Q\dyad{\bm0}\hat U_Q^\dagger}{\psi_g}\\
        +& \mel*{\psi_g}{\hat U_Q\hat U_Q^\dagger\hat{H}\hat U_Q\hat P_{\ket{\bar{\bm0}}}\hat U_Q^\dagger}{\psi_g},
    \end{aligned}
\end{equation}
where $\hat{P}_{\ket{\bar{\bm 0}}}$ projects onto the subspace orthogonal to $\ket{\bm 0}$.
This decomposition motivates the use of a projected sampler,
\begin{equation}
\hat{P}_{\ket{\bar{\bm 0}}} \hat{U}_Q^\dagger \ket{\tilde\psi_g},
\end{equation}
together with the basis state $\ket{\bm 0}$, to construct an extended sampled subspace. 
In other words, rather than sampling directly from $\ket{\tilde\psi_g'} = \hat{U}_Q^\dagger \ket{\tilde\psi_g}$, which is overly concentrated on $\ket{\bm 0}$, we use the projected state to explore the nontrivial basis states relevant for subspace expansion. 
For the numerical simulations presented in Sec.~\ref{sec:ising-model}, tensor-network methods are used to represent the Hamiltonians $\hat{H}$ and $\hat{H}'$, as well as the samplers $\ket{\tilde\psi_g}$ and $\hat{P}_{\ket{\bar{\bm 0}}}\hat{U}_Q^\dagger \ket{\tilde\psi_g}$.

For completeness, we conclude this section with three remarks on possible experimental implementations of the projected sampler. 
First, the projector $\hat P_{\ket{\bar{\bm0}}}$ may be embedded into a unitary through block-encoding techniques \cite{lin2022lecture} and appended to the circuit $\calC_Q^{\dagger}$.
Second, one may directly construct an ansatz $\ket{\Psi_{\rm s}}$ that approximates the normalized state $\hat P_{\ket{\bar{\bm0}}}\hat U_{Q}^{\dagger}\ket{\tilde\psi_g}$, in close analogy with Ref.~\cite{robledo2024chemistry_pub}, making the approach compatible with current quantum hardware.
Third, given an input state $\hat U_{\rm Q}^{\dagger}\ket{\tilde\psi_g}$ and a target output state $\hat P_{\ket{\bar{\bm0}}}\hat U_{\rm Q}^{\dagger}\ket{\tilde\psi_g}$, one may approximate the non-unitary projector by a unitary $\hat U_{{\ket{\bar{\bm0}}}}$, 
such that 
\begin{equation} \label{eq:U0}
\hat{P}_{\ket{\bar{\bm 0}}}\hat{U}_Q^\dagger\ket{\tilde\psi_g}
\approx
\hat{U}_{\ket{\bar{\bm 0}}}\hat{U}_Q^\dagger\ket{\tilde\psi_g},
\end{equation}
up to normalization.
The performance of this approximation is assessed in Appendix~\ref{appx:III}, and the corresponding quantum-hardware results are presented in Sec.~\ref{sec:expr}.

\section{Wavefunction Sparsity Measure}\label{sec:sparsity}

In this section, we introduce quantitative measures of wavefunction sparsity that are directly relevant to the resource requirements of SQD. 
Our goal is not merely to characterize sparsity itself, but to relate the probability distribution of measured computational basis states to the two key costs of SQD: the required sampled-subspace dimension $N_R$ and the number of measurement shots $N_S$. 
For this purpose, we employ the Lorenz curve and the Gini coefficient, which quantify how strongly the measurement distribution is concentrated on a small number of bitstrings.
These quantities allow us to derive bounds on the sampling overhead of standard SQD and, in turn, to clarify how FSQD can mitigate this overhead by enhancing wavefunction sparsity and reducing the resources required for subspace construction.

Suppose that the normalized quantum state being measured is written in the computational basis as 
\begin{equation}
\ket{\psi}=\sum_{I=1}^N c_I\ket{x_I}, 
\end{equation}
where $x_I\in\{0,1\}^n$ denotes a bitstring for an $n$-qubit system and $N=2^n$. 
Here the index $I\in\{1,\dots,N\}$ is chosen such that $I-1$ is the base-10 integer corresponding to the bitstring $x_I$. 
The vector 
\begin{equation}
    \vec{\psi}^{\circ2}\coloneqq (|c_1|^2,|c_2|^2,\dots,|c_I|^2,\dots,|c_N|^2)
\end{equation}
therefore represents the measurement probability distribution over computational basis states, and satisfies 
\begin{equation}
\sum_{I=1}^{N}|c_I|^2=1.
\end{equation}
An example of such a weight distribution is shown in Fig.~\ref{fig:sqd-filter}(c).

To characterize the concentration of this probability distribution, we sort the entries in non-decreasing order,
\begin{equation}
|c_1|^2\leqslant|c_2|^2\leqslant\dots\leqslant|c_N|^2.
\end{equation}
We then define the cumulative weight by 
\begin{equation}
\mathcal{L}\!\left(\frac{I}{N}\right)
\coloneqq
\sum_{J=1}^{I}|c_J|^2,
\qquad I=0,1,\dots,N,
\end{equation}
with the convention $\mathcal{L}(0)=0$, and regard $\mathcal{L}(x)$ for $x\in[0,1]$ as the piecewise-linear interpolation of these discrete values.
When plotted as a function of the normalized index $x=I/N$, $\mathcal{L}(x)$ defines the Lorenz curve, a standard tool for visualizing concentration and inequality~\cite{lorenz1905,Joseph1971Lorenz}.
In the original economic setting, $x$ represents the cumulative fraction of the population and $\mathcal{L}(x)$ the corresponding cumulative fraction of wealth. A more concentrated probability distribution corresponds to a Lorenz curve that deviates more strongly from the $45^\circ$ line (see Fig.~\ref{fig:lorenz-gini}).

Because the probabilities are sorted in non-decreasing order, the slope of $\mathcal{L}(x)$ is also non-decreasing.
More precisely, on the interval $x\in[(I-1)/N,I/N]$, one has
\begin{equation}
\mathcal{L}'(x)=N|c_I|^2,
\end{equation}
whenever the derivative exists.
Therefore, $\mathcal{L}(x)$ is a convex function.


A closely related and widely used scalar measure of concentration is the Gini coefficient~\cite{Gini1921,Dalton1920}.  
Graphically, it is given by twice the area between the Lorenz curve and the $45^{\circ}$ line, as illustrated in Fig.~\ref{fig:lorenz-gini}. 
Here, the $45^{\circ}$ line corresponds to the cumulative weight distribution of the equal superposition state $\ket{+}^{\otimes n}$, for which all computational basis states have equal probability. 
The Gini coefficient is given by~\cite{Hurley2009Gini}
\begin{equation}
    G = 1 - 2\sum_{I=1}^{N} \frac{|c_I|^2}{\Vert\vec{\psi}^{\circ2}\Vert_1}\left(\frac{N-I+\frac{1}{2}}{N}\right),
\end{equation}
where $\Vert\vec{\psi}^{\circ2}\Vert_1\equiv\sum_{I=1}^N |c_I|^2$ is the $\ell^1$ norm of $\vec{\psi}^{\circ2}$, and hence $\Vert\vec{\psi}^{\circ2}\Vert_1=1$ in our case because $\vec{\psi}^{\circ2}$ is a normalized probability vector.
The maximum value,
\begin{equation}\label{eq:most-sparse}
    \max(G) = 1-2^{-n},
\end{equation}
is attained for the sparsest possible distribution, in which all probability weight $|c_I|^2$ resides on a single computational basis state.
Thus, larger values of $G$ correspond to a more concentrated measurement distribution and hence to greater wavefunction sparsity. 
In contrast to the full Lorenz curve, the Gini coefficient provides a single scalar quantity that is particularly convenient for comparing different system sizes and for analyzing the scaling of sampling overhead.

\begin{figure}[tb]
    \includegraphics[width=\linewidth]{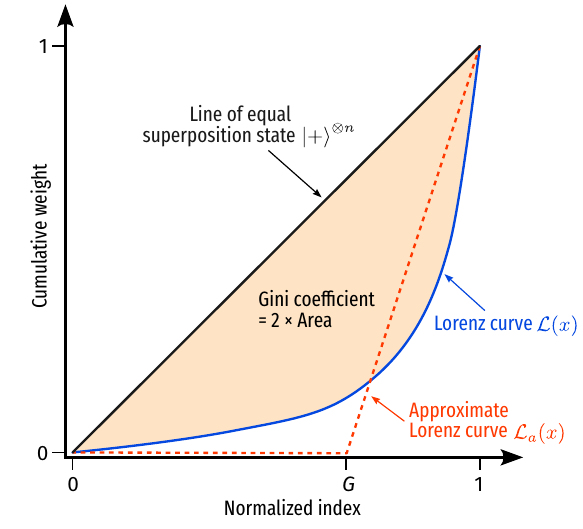}
    \caption{
      Schematic diagram of the Lorenz curve $\mathcal{L}(x)$ and the graphical interpretation of the Gini coefficient $G$. The diagonal line corresponds to the Lorenz curve for the equal-superposition state $\ket{+}^{\otimes n}$, and the approximate Lorenz curve $\mathcal{L}_a(x)$ is shown as a red dashed line. 
    }
    \label{fig:lorenz-gini}
\end{figure}


We also note that previous studies~\cite{robledo2024chemistry_pub,yu2025quantum} have used the sparsity measures $\alpha_L$ and $\beta_L$, defined by 
\begin{equation}
\sum_{J=N-L+1}^N |c_J|^2 \ge \alpha_L
\end{equation}
and
\begin{equation}
|c_{N-L+1}|^2 \ge \beta_L,
\end{equation}
to establish the convergence of the SQD algorithm. 
In fact, $\alpha_L$ is directly related to the cumulative weight $\calL(x)$ through 
\begin{equation}
    1-\mathcal{L}\!\left(\frac{I}{N}\right) = \sum_{J=I+1}^{N}|c_J|^2 \geqslant \alpha_{N-I}.
\end{equation}
The other sparsity measure $\beta_L$ is closely related to the threshold-based sparsity measures 
discussed in Ref.~\cite{Hurley2009Gini}, 
which are generally less convenient to optimize and compare across different system sizes.
By contrast, the Gini coefficient provides a more robust and informative global measure of concentration~\cite{Hurley2009Gini}. 
In particular, because it summarizes the full probability distribution into a single scalar with a clear geometric interpretation, it is well suited to analyzing the scaling of SQD resource requirements and the improvement achieved by FSQD.

The Lorenz curve directly yields bounds on the resources required for SQD, namely the sampled-subspace dimension needed to achieve a target accuracy and the number of measurement shots required to obtain that subspace with high probability. 
To make this precise, let $E_{\mathrm{SQD}}$ denote the lowest eigenvalue of the truncated Hamiltonian constructed in the sampled subspace, and define the SQD deviation by
\begin{equation}
    \Delta_{\mathrm{SQD}}
    \coloneqq
    E_{\mathrm{SQD}}-\expval*{\hat H}{\psi}.
\end{equation}
The following theorem provides sufficient conditions to guarantee $\Delta_{\mathrm{SQD}}\le \varepsilon$. 
\begin{theorem}\label{theo:sqd}
    For a quantum state $\ket{\psi}$ and a Hamiltonian $\hat H$ of an $n$-qubit system, 
    a sufficient sampled-subspace dimension for guaranteeing 
    \begin{equation}
        E_{\mathrm{SQD}}-\expval*{\hat H}{\psi}\le \varepsilon
    \end{equation}
    is
    \begin{equation}\label{eq:N_R-bound}
        N_R = \left(1-\calL^{-1}(\tilde\varepsilon^2)\right)N,
    \end{equation}
    and the number of shots required to obtain these bitstrings with probability at least $1-\eta$ is
    \begin{equation}\label{eq:N_S-bound}
        N_S = \frac{\log[(1-\calL^{-1}(\tilde\varepsilon^2))N/\eta]}{\calL'(\calL^{-1}(\tilde\varepsilon^2))/N},\qquad  \eta>0.
    \end{equation}
\end{theorem}
Here $\calL^{-1}$ denotes the inverse function of the Lorenz curve. We have defined $\tilde{\varepsilon} \coloneqq \varepsilon/{\left(2\sqrt{2}\rho(\hat H)\right)}$, where $\rho(\hat H)$ denotes the spectral norm of $\hat H$, i.e., the largest singular value of the matrix $H$. The proof of Theorem~\ref{theo:sqd} is given in Appendix~\ref{appx:II}.

We emphasize that $\ket{\psi}$ in Theorem~\ref{theo:sqd} need not be the exact ground state $\ket{\psi_g}$ of $\hat H$; it may be any normalized sampler state.
In the special case $\ket{\psi}=\ket{\psi_g}$, one has $\expval*{\hat H}{\psi_g}=E_0$, so that the theorem directly provides a bound on the resources required for SQD to approximate the true ground-state energy.
Moreover, in this case $E_{\mathrm{SQD}}-\expval*{\hat H}{\psi_g}\ge 0$ by the variational principle, and the deviation $\Delta_{\mathrm{SQD}}$ reduces to the usual energy-estimation error.

Although Theorem~\ref{theo:sqd} is expressed directly in terms of the Lorenz curve, obtaining an analytical expression for $\mathcal{L}(x)$ is generally difficult.
To obtain simpler bounds that are easier to analyze, we instead make use of the Gini coefficient $G$.
Thanks to its graphical interpretation, $G$ yields lower bounds on both $N_R$ and $N_S$.
Moreover, unlike the full Lorenz curve, $G$ is a single scalar quantity and is therefore particularly convenient for analyzing finite-size scaling. 
To this end, we approximate the Lorenz curve by the piecewise-linear function $\mathcal{L}_a(x)$ defined by
\begin{equation}
\mathcal{L}_a(x)=
\begin{cases}
0, & x<G,\\[4pt]
\dfrac{x-G}{1-G}, & x\geqslant G,
\end{cases}
\end{equation}
as illustrated in Fig.~\ref{fig:lorenz-gini}. 
The following corollary then gives simple lower bounds on the SQD resource requirements. 
\begin{corollary}\label{coro:sqd}
    For $\tilde{\varepsilon}^2$ smaller than the value at the intersection point of the true Lorenz curve $\calL(x)$ and its approximation $\calL_a(x)$, the sampled-subspace dimension and the number of shots satisfy 
    \begin{equation}
    N_R \geqslant (1-G)N(1-\tilde{\varepsilon}^2),
    \end{equation}
    and
    \begin{equation}
    N_S 
    \geqslant (1-G)N(1-\tilde{\varepsilon}^2)\log\!\left[\frac{(1-G)N(1-\tilde{\varepsilon}^2)}{\eta}\right], 
    \label{eq:Ns_bound}
    \end{equation}
    respectively.
%
\end{corollary}
\begin{proof}
    From the graphical construction in Fig.~\ref{fig:lorenz-gini}, one has $\calL^{-1}(\tilde{\varepsilon}^2)\leqslant \calL_a^{-1}(\tilde{\varepsilon}^2)$. 
    Substituting this into Eq.~\eqref{eq:N_R-bound}, we obtain 
    \begin{equation}
    N_R = \left(1-\calL^{-1}(\tilde{\varepsilon}^2)\right)N \geqslant \left(1-\calL_a^{-1}(\tilde{\varepsilon}^2)\right)N .
    \end{equation}
    Since
    \begin{equation}
    \mathcal{L}_a^{-1}(\tilde{\varepsilon}^2)=G+(1-G)\tilde{\varepsilon}^2,
    \end{equation}
    it follows that
    \begin{equation}
    1-\mathcal{L}_a^{-1}(\tilde{\varepsilon}^2)
    =(1-G)(1-\tilde{\varepsilon}^2),
    \end{equation}
    and hence
    \begin{equation}
    N_R \geqslant (1-G)N(1-\tilde{\varepsilon}^2).
    \end{equation}
    
    Next, because $\mathcal{L}(x)$ is convex, its derivative is non-decreasing, and therefore
    \begin{equation}
    \mathcal{L}'\!\left(\mathcal{L}^{-1}(\tilde{\varepsilon}^2)\right)
    \leqslant
    \mathcal{L}'\!\left(\mathcal{L}_a^{-1}(\tilde{\varepsilon}^2)\right).
    \end{equation}
    To bound the latter quantity, let
    \begin{equation}
    x_0\coloneqq \mathcal{L}_a^{-1}(\tilde{\varepsilon}^2).
    \end{equation}
    Since $\mathcal{L}(x)$ is convex and satisfies $\mathcal{L}(1)=1$, the tangent line at $x=x_0$ lies below the point $(1,1)$, implying
    \begin{equation}
    \mathcal{L}'(x_0)
    \leqslant
    \frac{1-\mathcal{L}(x_0)}{1-x_0}
    \leqslant
    \frac{1}{1-x_0}.
    \end{equation}
    Using
    \begin{equation}
    1-x_0 = 1-\mathcal{L}_a^{-1}(\tilde{\varepsilon}^2)
    =(1-G)(1-\tilde{\varepsilon}^2),
    \end{equation}
    we obtain
    \begin{equation}
    \mathcal{L}'\!\left(\mathcal{L}_a^{-1}(\tilde{\varepsilon}^2)\right)
    \leqslant
    \frac{1}{(1-G)(1-\tilde{\varepsilon}^2)}.
    \end{equation}
    Combining this with Eq.~\eqref{eq:N_S-bound}, we find
    \begin{equation}
    N_S
    \geqslant
    (1-G)N(1-\tilde{\varepsilon}^2)
    \log\!\left[\frac{(1-G)N(1-\tilde{\varepsilon}^2)}{\eta}\right].
    \end{equation}
    %
\end{proof}

We emphasize that the two bounds in Corollary~\ref{coro:sqd} are lower bounds on the SQD resource requirements and therefore do not by themselves guarantee the target accuracy $\varepsilon$, in contrast to Eqs.~\eqref{eq:N_R-bound} and \eqref{eq:N_S-bound}.
Nevertheless, they are highly useful for understanding the scaling of SQD overhead with system size.
For example, if the $N_R$ relevant bitstrings were sampled uniformly, the expected number of shots required to observe each of them at least once would be $N_RK_{N_R}$, where $K_{N_R}=\sum_{k=1}^{N_R}\frac{1}{k}$ is the $N_R$th harmonic number and satisfies $K_{N_R}\sim \log N_R$ for large $N_R$. 
This follows from the coupon collector problem~\cite{ferrante2012notecoupon}: once $k-1$ distinct bitstrings have been observed, the expected number of additional shots needed to obtain a new one is $N_R/(N_R-k+1)$, and summing over $k=1,\dots,N_R$ yields $N_RK_{N_R}$.  
This has the same leading-order scaling as the lower bound on $N_S$ above in Eq.~\eqref{eq:Ns_bound}. 
This observation suggests that an ideal sampler should distribute probability approximately uniformly over the $N_R$ basis states relevant for subspace construction.
At the same time, such an optimal uniform sampler contains no information about the relative phases and amplitudes of the target ground state.
Therefore, efficient preparation of an optimal sampler does not by itself solve the ground-state preparation problem, and the latter is not reducible to the former.



When the wavefunction amplitudes $|c_I|^2$ follow exponential or power-law decay, the corresponding Lorenz curve can be derived analytically in terms of a rate parameter $\lambda$, which characterizes the concentration of the measurement distribution and, correspondingly, the Gini coefficient (see Appendix~\ref{appx:II}).
If $\lambda$ is independent of the system size $n$, as assumed in Ref.~\cite{robledo2024chemistry_pub}, then the Gini coefficient (as well as the Lorenz curve itself) satisfies $(1-G)N\propto\mathrm{poly}(n)$, implying that the SQD resource requirements, namely $N_S$ and $N_R$, scale only polynomially with system size.
In practice, however, one often finds the scaling
\begin{equation}
(1-G)N \propto 2^{gn+c},
\end{equation}
where $g$ characterizes the effective system-size exponent and $c$ is a size-independent constant.
Equivalently, this may be written as
\begin{equation}
1-G \propto 2^{-(1-g)n+c}.
\end{equation}
When $g>0$, the quantity $(1-G)N$ grows exponentially with system size, leading to an exponentially large sampling overhead.

%
Moreover, reducing the sufficient subspace dimension by a factor $k$ generally comes at the cost of reduced energy-estimation accuracy. 
For example, setting $N_R\mapsto N_R/k$ suggests the transformations $\varepsilon\mapsto\varepsilon^{1/k}$ for exponential decay, and $\varepsilon\mapsto\varepsilon \sqrt{k}$ for second-order power-law decay in the regime $\tilde{\varepsilon}\ll 1$. 
These relations follow by inverting the asymptotic dependence of the sufficient subspace dimension $N_R$ on $\tilde\varepsilon$: for exponential decay, $N_R\propto \log(1/\tilde\varepsilon^2)$, whereas for second-order power-law decay, $N_R\propto \tilde\varepsilon^{-2}$ in the regime $\tilde\varepsilon\ll 1$ (see Appendix~\ref{appx:II}). 
This illustrates that, within standard SQD, lowering the sampling cost by simply reducing the sampled-subspace dimension necessarily worsens the attainable accuracy. 
%

The role of FSQD is to change this scaling behavior at the level of the measurement distribution itself.
By applying a quantum filter, FSQD aims to transform the original ground-state wavefunction into a filtered one whose probability distribution is substantially more sparse in the computational basis. 
Correspondingly, the measurement distribution of the filtered sampler is expected to become more sparse than that of the original sampler that approximates the target ground state. 
Ideally, this transformation produces a system-size-independent Gini coefficient, or at least weakens its system-size dependence, corresponding to a smaller value of $g$ in the scaling relation above.
In this way, FSQD can reduce the sampling overhead not merely by shrinking the sampled subspace at fixed distribution, but by reshaping the distribution itself so that the basis states relevant for subspace construction are sampled much more efficiently.


\section{Quantum circuit filter encoding}\label{sec:automatic}

\begin{figure}[tb]
    \includegraphics[width=\linewidth]{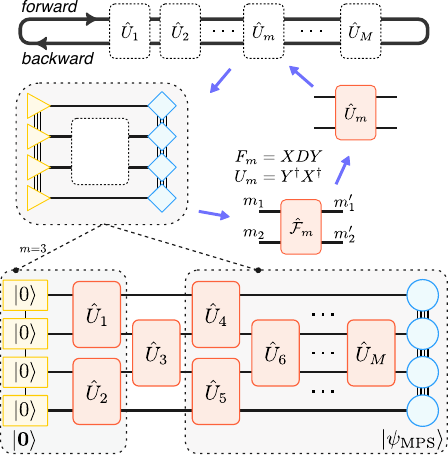}
    \caption{    
      Schematic illustration of the forward and backward sweeps in the MPS-circuit-encoding algorithm. The top panel shows the bidirectional optimization over the local two-qubit unitary matrices $U_m$. 
      The middle-left panel illustrates the local tensor-network environment for optimizing the $m$th unitary. 
      The middle-right panel shows the fidelity tensor $\hat{\mathcal F}_m$ and its matrix representation $F_m$, from which the locally optimal unitary matrix $U_m=Y^\dagger X^\dagger$ is obtained through the singular-value decomposition $F_m=XDY$. 
      The lower panel shows the circuit $\mathcal{C}_{\rm MPS}$, composed of local two-qubit unitaries $\hat U_m$, which is optimized so that the all-zero product state $\ket{\bm 0}$ is mapped approximately to the target MPS $\ket{\psi_{\rm MPS}}$. 
      The gray-shaded regions surrounded by dashed lines represent the contracted states $\ket{\Psi_m^{L}}$ and $\bra{\Psi_m^{R}}$.
    }
    \label{fig:tn-alg}
\end{figure}

\begin{figure}[tb]
\begin{algorithm}[H]
\caption{MPS circuit encoding~\cite{Shirakawa2025}}
\label{alg:two}
\begin{algorithmic}[1]
    \Require Target MPS $\ket{\psi_{\rm MPS}}$, initial MPS $\ket{\bm 0}$, and circuit structure of $\mathcal{C}_{\rm MPS}$
    \Ensure Local two-qubit unitaries $\hat U_1,\dots,\hat U_m,\dots,\hat U_M$
    \Function{\scshape Bidirectional}{\emph{direction}}
        \If{\emph{direction} \textbf{is} \emph{forward}}
            \State $a,b,s \gets 1,M,+1$
        \ElsIf{\emph{direction} \textbf{is} \emph{backward}}
            \State $a,b,s \gets M,1,-1$
        \EndIf
        \For{$m=a$ \textbf{to} $b$ \textbf{with step} $s$}
            \State $\hat{\mathcal F}_m \gets \tr_{\bar m}(\ket{\Psi_m^{L}}\bra{\Psi_m^{R}})$ \Comment{Eq.~\eqref{eq:fidelity_tensor}}
            \State $F_m \gets$ matrix representation of $\hat{\mathcal F}_m$
            \State $X,D,Y \gets \mathrm{svd}(F_m)$ \Comment{perform SVD}
            \State $U_m \gets Y^{\dagger}X^{\dagger}$ \Comment{Eq.~\eqref{eq:optimal_u}}
            \State $\hat U_m \gets$ operator/tensor representation of $U_m$
            \If{\emph{direction} \textbf{is} \emph{forward}}
                \State $\ket{\Psi_{m+1}^{L}} \gets \hat U_m \ket{\Psi_m^{L}}$ if $m\neq b$
            \ElsIf{\emph{direction} \textbf{is} \emph{backward}}
                \State $\bra{\Psi_{m-1}^{R}} \gets \bra{\Psi_m^{R}} \hat U_m$ if $m\neq b$
            \EndIf
        \EndFor
    \EndFunction
    \State $\hat U_m \gets \hat I$ for $m=1,\dots,M$
    \State $\ket{\Psi_1^{L}} \gets \ket{\bm 0}$
    \State $\bra{\Psi_M^{R}} \gets \bra{\psi_{\rm MPS}}$
    \For{$i=1$ \textbf{to} $N_{\rm iter}$}
        \State \textsc{Bidirectional}(\emph{forward})
        \State \textsc{Bidirectional}(\emph{backward})
    \EndFor
\end{algorithmic}
\end{algorithm}
\end{figure}

To construct a quantum circuit filter, 
we first approximate the ground state of $\hat H$ by the density matrix renormalization group (DMRG) method 
and obtain its matrix product state (MPS) representation, denoted by $\ket{\psi_{\rm MPS}}$. 
We then use the automatic MPS circuit encoding method introduced in Ref.~\cite{Shirakawa2025} to construct a quantum circuit approximating this target MPS.
In the present work, this encoded circuit provides the filter circuit used in FSQD.
More specifically, the goal is to optimize a circuit $\mathcal{C}_{\rm MPS}$ whose associated unitary $\hat U_{\rm MPS}$ satisfies 
\begin{equation}
    \hat U_{\rm MPS}\ket{\bm 0}\approx \ket{\psi_{\rm MPS}},
\end{equation}
with controllable fidelity. 
Related approaches to encoding tensor-network states into quantum circuits have recently attracted considerable interest~\cite{Ran2020Encoding,Malz2024Preparation,Smith2024Constant}.

Assume that the circuit $\mathcal{C}_{\rm MPS}$ is composed of a sequence of local two-qubit unitaries $\hat U_m$, each acting on qubits $m_1$ and $m_2$.
The total unitary of the circuit is
\begin{equation}
\hat U_{\rm MPS} = \hat U_M \cdots \hat U_{m+1}\hat U_m \cdots \hat U_1.
\end{equation}
Each local unitary $\hat U_m$ may be viewed either as a $4\times 4$ matrix $U_m$ or, equivalently, as a four-leg tensor with indices $m_1',m_2',m_1,m_2\in\{0,1\}$.
The full circuit $\mathcal{C}_{\rm MPS}$ is therefore represented in tensor-network form by applying these local tensors to the MPS of the all-zero state $|\bm 0\rangle$, as illustrated in Fig.~\ref{fig:tn-alg}. 
Here, we assume that these two-qubit unitaries are arranged in a brick-wall form, and therefore the circuit structure itself is not optimized. 
The number of such layers plays a role analogous to the bond dimension in DMRG and controls the fidelity of the encoded circuit.

The pseudocode of the encoding procedure is given in Algorithm~\ref{alg:two}.
Its overall structure is as follows.
First, we initialize the target MPS $\ket{\psi_{\rm MPS}}$, the initial MPS $\ket{\bm 0}$, and the circuit structure of $\mathcal{C}_{\rm MPS}$.
Second, we define two sets of MPS,
\(\{\ket{\Psi_m^{L}}\}\) and \(\{\bra{\Psi_m^{R}}\}\), by
\begin{eqnarray}
\ket{\Psi_m^{L}} &=& \hat U_{m-1}\cdots \hat U_1\ket{\bm 0},\\
\bra{\Psi_m^{R}} &=& \bra{\psi_{\rm MPS}}\hat U_M\cdots \hat U_{m+1}.
\end{eqnarray}
That is, the local unitaries with indices \(i<m\) and \(i>m\) are contracted with $\ket{\bm 0}$ and $\ket{\psi_{\rm MPS}}$, respectively.
An example with $m=3$ is shown in Fig.~\ref{fig:tn-alg}, where the contracted states $\ket{\Psi_3^{L}}$ and $\bra{\Psi_3^{R}}$ appear as the gray-shaded regions surrounded by dashed lines.

Following Ref.~\cite{Shirakawa2025}, we define the fidelity tensor
\begin{equation}\label{eq:fidelity_tensor}
\hat{\mathcal{F}}_m = \tr_{\bar m}(\ket{\Psi_m^{L}}\bra{\Psi_m^{R}}),
\end{equation}
where $\tr_{\bar m}$ denotes contraction over all MPS site indices except those associated with qubits $m_1$ and $m_2$.
The resulting object has four indices, corresponding to a local two-qubit tensor.
The objective function to be optimized is
\begin{equation}\label{eq:mps-objective-func}
f_m = \tr_m(\hat{\mathcal{F}}_m \hat U_m)=\tr(F_m U_m),
\end{equation}
where $F_m$ and $U_m$ denote the $4\times 4$ matrix representations of $\hat{\mathcal{F}}_m$ and the local two-qubit unitary $\hat U_m$, respectively.
Note that $F_m$ is generally neither Hermitian nor unitary~\cite{Shirakawa2025}.

Performing the singular-value decomposition
\begin{equation}\label{eq:mps-Fm}
F_m = X D Y,
\end{equation}
with $X$ and $Y$ unitary and $D$ diagonal with singular values $d_i$ \((i=1,2,3,4)\), we obtain
\begin{equation}
f_m = \tr(XDYU_m)=\tr(DZ)=\sum_{i=1}^{4} d_i[Z]_{ii},
\end{equation}
where
\begin{equation}
Z = YU_mX 
\end{equation}
is unitary.
For a unitary matrix, one has
\begin{equation}
\sum_{i'} |[Z]_{ii'}|^2 = 1,
\end{equation}
and hence \(|[Z]_{ii}|\le 1\).
Therefore,
\begin{equation}
|f_m|
=
\left|\sum_{i=1}^{4} d_i[Z]_{ii}\right|
\le
\sum_{i=1}^{4} d_i|[Z]_{ii}|
\le
\sum_{i=1}^{4} d_i,
\end{equation}
where equality is attained when \(Z=I\) up to an irrelevant global phase.
Thus, the locally optimal two-qubit unitary matrix is
\begin{equation}\label{eq:optimal_u} 
U_m = Y^\dagger X^\dagger,
\end{equation}
which maximizes \(|f_m|\).

Each iteration of Algorithm~\ref{alg:two} consists of a forward sweep and a backward sweep. 
In the forward sweep, the locally optimal unitaries \(U_m\) are updated sequentially from \(m=1\) to \(m=M\), while the states \(\{\ket{\Psi_m^{L}}\}\) are updated accordingly.
In the backward sweep, the procedure is reversed from \(m=M\) down to \(m=1\), updating the states \(\{\bra{\Psi_m^{R}}\}\).
The iteration is repeated either for a fixed number of steps or until a prescribed convergence criterion for \(f_m\) is reached. 
The optimized unitary matrices $\{U_m\}$ are subsequently decomposed into an elementary gate set consisting of single- and two-qubit gates (see Appendix~\ref{appx:IV}).

\section{Numerical benchmark for the quantum Ising model}\label{sec:ising-model}

\begin{figure*}[tbh]
    \includegraphics[width=\linewidth]{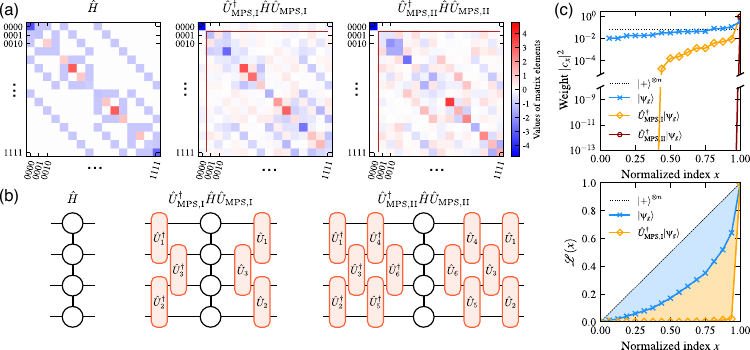}
    \caption{
      (a) Matrix representations of the original Hamiltonian $\hat H$ [Eq.~\eqref{eq:ising-model}] and the filtered Hamiltonian $\hat H'=\hat U_{\rm MPS,I(II)}^{\dagger}\hat H\hat U_{\rm MPS,I(II)}$ for one-layer and two-layer filters on a four-site system. 
      (b) Corresponding tensor-network, i.e., MPO, representations. 
      (c) Sorted weights $|c_I|^2$ and cumulative weights $\calL(x)$ in the computational basis for the same four-site system. 
      The Gini coefficient equals twice the shaded area between the  Lorenz curve and the diagonal line corresponding to the equal-superposition state $|+\rangle^{\otimes n}$. 
      Examples are shown for $\ket{\psi_g}$ (blue), $\umpsdI\ket{\psi_g}$ (yellow), and $\umpsdII\ket{\psi_g}$ (brown). 
      For the two-layer filter, the weight distribution is so strongly concentrated on $\ket{\bm 0}$ that its Lorenz curve is almost indistinguishable from a delta-function-like contribution at $x=1$ on this scale (not shown). 
    }
    \label{fig:sqd-filter}
\end{figure*}

\begin{figure*}[tb]
    \includegraphics[width=\linewidth]{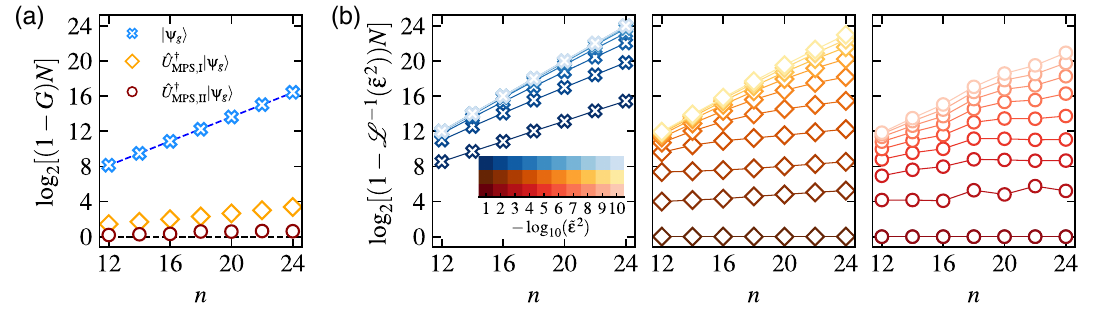}
    \caption{
      (a) Deviation of the Gini coefficient from unity, $1-G$, as a measure of wavefunction sparsity, plotted as a function of system size $n$. 
      A fit to the ground state $\ket{\psi_g}$ of the original Hamiltonian $\hat H$ gives $(1-G)N \sim 2^{0.7n}$, shown by the blue dashed line, whereas the filtered states generated by $\hat U_{\rm MPS,I}$ and $\hat U_{\rm MPS,II}$ approach the sparsest-case scaling, shown by the black dashed line. 
      (b) Inverse Lorenz curve as functions of $n$ for $|\psi_g\rangle$ (left), $\hat U_{\rm MPS,I}^\dagger|\psi_g\rangle$ (middle), and $\hat U_{\rm MPS,II}^\dagger|\psi_g\rangle$ (right). 
      The color intensity corresponds to the target estimation error $\tilde{\varepsilon}^2=10^{-1}$, $10^{-2}$, $\dots$, and $10^{-10}$, as indicated in the inset of the left panel. 
      }
    \label{fig:sqd-sparsity}
\end{figure*}

To benchmark the proposed FSQD protocol, we consider the quantum Ising chain in both transverse and longitudinal fields, described by the Hamiltonian
\begin{equation}\label{eq:ising-model}
    \hat{H} = \hat{H}_0 + \hat{V},
\end{equation}
where
\begin{equation}
    \hat{H}_0 = -J\sum_{i}\hat\sigma_{i}^{z}\hat\sigma_{i+1}^{z} - h^x\sum_{i}\hat\sigma_{i}^{x},
\end{equation}
and
\begin{equation}
    \hat{V} = -h^z\sum_{i}\hat\sigma_{i}^{z}.
\end{equation}
Here, $\hat{\sigma}_x$ and $\hat{\sigma}_z$ are Pauli operators acting on site $i$. 
The term $\hat{H}_0$ is the transverse-field Ising model with ferromagnetic nearest-neighbor coupling $J>0$~\cite{PFEUTY197079}. 
The transverse field favors a paramagnetic state, while the Ising interaction favors ferromagnetic order. 
The model $\hat{H}_0$ is exactly solvable via a Jordan-Wigner transformation, which maps the spin system into noninteracting fermions~\cite{LIEB1961407}. 
Its ground state exhibits a quantum phase transition between ferromagnetic and paramagnetic phases at the critical point $h_c^x=J$~\cite{sachdev1999quantum}.
Near this critical point, the longitudinal field $\hat{V}$ breaks integrability and gives rise to nontrivial spin-excitation dynamics of considerable theoretical and experimental interest~\cite{Coldea2010}.

Throughout this section, we consider an ideal numerical benchmark of FSQD. 
In particular, we assume that the sampler state $|\tilde \psi_g\rangle$ is given by the ground state $|\psi_g\rangle$ of the Hamiltonian $\hat H$, obtained exactly by direct diagonalization for systems with $n\leqslant24$ and to essentially exact accuracy by DMRG for larger systems. 
No hardware noise is included in the these numerical simulations. 
This setting allows us to isolate the intrinsic effect of filter-assisted subspace construction from imperfections associated with actual quantum devices.

We compare the following sampling strategies.
First, as a reference, we consider the standard SQD protocol, in which bitstrings are sampled directly from the ground state $\ket{\psi_g}$ of the original Hamiltonian $\hat H$. 
Second, we consider FSQD using two different MPS-circuit filters, denoted by $\hat U_{\rm MPS,I}$ and $\hat U_{\rm MPS,II}$, corresponding to one and two layers of local two-qubit unitaries, respectively [see Fig.~\ref{fig:sqd-filter}(b)]. 
The corresponding filtered states are
$\hat U_{\rm MPS,I}^{\dagger}\ket{\psi_g}$ and
$\hat U_{\rm MPS,II}^{\dagger}\ket{\psi_g}$.
Since the filter is designed to concentrate the dominant component of the ground state onto the computational basis state $\ket{\bm 0}$, we use the projected filtered sampler states
$\hat P_{\ket{\bar{\bm0}}}\hat U_{\rm MPS,I}^{\dagger}\ket{\psi_g}$
and
$\hat P_{\ket{\bar{\bm0}}}\hat U_{\rm MPS,II}^{\dagger}\ket{\psi_g}$
for the FSQD energy-estimation benchmark. 
Here, $\hat P_{\ket{\bar{\bm0}}}$ denotes the projection onto the subspace orthogonal to $\ket{\bm0}$, as introduced in Sec.~\ref{sec:fsqd}. 
In the numerical simulations, this projector is implemented exactly and therefore introduces no additional approximation, such as the unitary approximation $\hat U_{{\ket{\bar{\bm0}}}}$ used in Eq.~(\ref{eq:U0}). 

In the simulations, we set $J=1$, $h^x=1$, and $h^z=0.05$, i.e., in a regime close to quantum criticality, and impose open boundary conditions. 
All DMRG and tensor-network computations are performed using the \textsc{Julia}~\cite{Julia2017} package \emph{ITensors.jl}~\cite{ITensor}. 
The ground state $\ket{\psi_g}$ is obtained by high-accuracy DMRG. For system sizes up to $n=100$, a maximum bond dimension of $20$ is sufficient to converge the ground-state energy and wavefunction to the accuracy required in the present benchmark. 
The MPS-circuit filters $\hat U_{\rm MPS,I}$ and $\hat U_{\rm MPS,II}$ are constructed from this ground-state MPS using the circuit-encoding algorithm described in Sec.~\ref{sec:automatic}. 
For tensor-network calculations involving the filtered Hamiltonian
\begin{equation}
\hat H' = \hat U_{\rm MPS,\alpha}^{\dagger}\hat H \hat U_{\rm MPS,\alpha},  
\end{equation}
with $\alpha=\mathrm{I},\mathrm{II}$, we use its MPO representation and set the maximum bond dimension to $50$. 

Figure~\ref{fig:sqd-filter}(a) compares the matrix representations of the original Hamiltonian $\hat H$ and the filtered Hamiltonian $\hat H'$ for a four-site system. 
In the filtered Hamiltonian, the dominant matrix element is shifted to the upper-left corner, corresponding to the computational basis state $\ket{0000}$. 
Moreover, the off-diagonal elements in the row and column associated with $\ket{0000}$ are already suppressed by one-layer filter $\hat U_{\rm MPS,I}$ and are reduced even further by the two-layer filter $\hat U_{\rm MPS,II}$. 
This behavior reflects the fact that increasing the number of filter layers improves the circuit approximation to the ground state, thereby making the filtered ground state more strongly concentrated on $\ket{0000}$. 
At the same time, away from the dominant $\ket{0000}$ sector, the filtered Hamiltonian becomes progressively denser than the original sparse Hamiltonian. 
The corresponding MPO representations are illustrated in Fig.~\ref{fig:sqd-filter}(b).

Figure~\ref{fig:sqd-filter}(c) shows the weight distributions of the original and filtered ground states in the computational basis. 
Hereafter, we denote by $\calC_{\rm MPS,I}$ and $\calC_{\rm MPS,II}$ the MPS-circuit filters composed of one and two layers of local two-qubit unitaries arranged in a brick-wall form, respectively, and by $\hat U_{\rm MPS,I}$ and $\hat U_{\rm MPS,II}$ the corresponding unitary operators [see Fig.~\ref{fig:sqd-filter}(b)]. 
Although the small system size obscures the asymptotic behavior of the weight distribution in the original ground state $\ket{\psi_g}$, the filtered states 
$\hat U_{\rm MPS,I}^{\dagger}\ket{\psi_g}$ and
$\hat U_{\rm MPS,II}^{\dagger}\ket{\psi_g}$
are strongly concentrated on $\ket{0000}$, with the weights of other computational-basis states suppressed by several orders of magnitude. 
Moreover, as expected, the two-layer filter produces a larger weight on $\ket{0000}$ than the one-layer filter, 
indicating that $\calC_{\rm MPS,II}$ provides a more accurate circuit encoding of $\ket{\psi_g}$. 
A more detailed characterization of the dependence on the number of layers is given in Appendix~\ref{appx:III}.


As discussed in Sec.~\ref{sec:sparsity}, the Gini coefficient of a quantum state, especially its scaling with system size, provides direct insight into the resource requirements of SQD.
We therefore compute the Gini coefficients $G$ of $\ket{\psi_g}$, $\umpsdI\ket{\psi_g}$, and $\umpsdII\ket{\psi_g}$ for systems with $n\leqslant24$, where the full state vectors can be stored explicitly in memory. 
Figure~\ref{fig:sqd-sparsity}(a) plots the deviation $1-G$ as a function of system size $n$, together with fits of the form 
\begin{equation}
    \log_2[(1-G)N] = gn + c. 
\end{equation}
For the ground state $\ket{\psi_g}$ of the original Hamiltonian, we obtain $g\simeq 0.7$, corresponding to the scaling $(1-G)N\sim 2^{0.7 n}$. 
According to Corollary~\ref{coro:sqd}, this implies that both the required subspace dimension $N_R$ and the number of shots $N_S$ in standard SQD grow at least exponentially with system size. 
By contrast, $\log_2[(1-G)N]$ for $\umpsdI\ket{\psi_g}$ exhibits a much weaker system-size dependence, while $\umpsdII\ket{\psi_g}$ closely follows the sparsest-case behavior, $G=1-2^{-n}$. 
This demonstrates that the circuit filter substantially enhances the sparsity of the measurement distribution and therefore strongly reduces the lower bounds on the SQD resources. 


We emphasize, however, that the lower bounds derived from Corollary~\ref{coro:sqd} do not by themselves guarantee that a desired accuracy can be achieved with a small, size-independent number of samples. 
To estimate the sufficient resource requirements more directly, one must analyze the Lorenz curve itself. 
This is shown in Fig.~\ref{fig:sqd-sparsity}(b), where we examine the scaling of $(1-\calL^{-1}(\tilde{\varepsilon}^2))N$ with system size $n$ for different target energy-estimation errors $\tilde{\varepsilon}$. 
For standard SQD, the scaling follows the same general trend as the Gini-coefficient analysis. 
As $\tilde\varepsilon$ decreases, the fitted parameter $g$ rapidly approaches one, meaning that the sufficient subspace dimension and number of shots approach the full Hilbert-space dimension, as discussed in Theorem~\ref{theo:sqd}. 
By contrast, for the filtered states
$\hat U_{\rm MPS,I}^{\dagger}\ket{\psi_g}$ and $\hat U_{\rm MPS,II}^{\dagger}\ket{\psi_g}$,
the fitted parameter $g$ approaches the sparsest-case value $g=0$ as $n$ increases for a broad range of target accuracies, namely $\tilde{\varepsilon}^2\leqslant 10^{-2}$ and $\tilde{\varepsilon}^2\leqslant 10^{-6}$, respectively. 
This indicates that, for the filter-assisted sampler, the resource requirements $N_R$ and $N_S$ exhibit only weak system-size dependence. 
Nevertheless, the offset $c$ in the scaling relation may still lead to a large absolute subspace dimension when extremely high accuracy, i.e., very small $\tilde\varepsilon$, is desired.

We next apply the standard SQD and FSQD protocols to estimate the ground-state energy of the quantum Ising model in Eq.~\eqref{eq:ising-model}. 
For the standard SQD calculations, the sampled subspace is constructed by drawing bitstrings directly from the ground state $\ket{\psi_g}$ of the original Hamiltonian $\hat H$. 
For FSQD, we first apply the circuit filter to obtain 
\begin{equation}\label{eq:psip}
\ket{\psi_g'} = \hat U_{\rm MPS}^{\dagger}\ket{\psi_g}, 
\end{equation}
where $\hat U_{\rm MPS}$ denotes either $\hat U_{\rm MPS,I}$ or $\hat U_{\rm MPS,II}$. 
Because the filter is designed to concentrate the dominant component of the ground state onto $\ket{\bm0}$, repeated sampling of $\ket{\bm0}$ does not efficiently enlarge the sampled subspace. 
Therefore, following the protocol of Sec.~\ref{sec:fsqd}, we remove the $\ket{\bm0}$ component and sample from the normalized projected state
\begin{equation}
\ket{\bar{\bm0}}_g\coloneqq\frac{
\hat P_{\ket{\bar{\bm0}}}\hat U_{\rm MPS}^{\dagger}\ket{\psi_g}
}{
\left\|
\hat P_{\ket{\bar{\bm0}}}\hat U_{\rm MPS}^{\dagger}\ket{\psi_g}
\right\|_2
}.\label{eq:0_bar}
\end{equation}
In practice, this projected state is constructed numerically by contracting the ground-state MPS $|\psi_g\rangle$ with the circuit filter $\hat U_{\rm MPS}$. Equivalently, we may obtain the ground-state MPS $\ket{\psi_g'} $ of the filtered Hamiltonian $\hat H'\coloneqq\hat U_{\rm MPS}^\dagger\hat H\hat U_{\rm MPS}$ directly. 
In either case, we then evaluate the overlap with $\ket{\bm 0}$, 
subtract the corresponding $\ket{\bm0}$ component, 
and normalize the resulting MPS. 
After normalization, we draw $N_S$ samples from the projected sampler state and retain the $N_R$ most frequent bitstrings.


Given the selected $N_R$ bitstrings, we construct the truncated Hamiltonian matrix in the sampled subspace. 
For FSQD, the relevant matrix elements are
$\mel{x_I}{\hat H'}{x_J}$,
where $\ket{x_I}$ and $\ket{x_J}$ are product-state MPSs corresponding to the selected bitstrings. 
These matrix elements are evaluated by standard MPO--MPS contractions.
The construction for the original Hamiltonian $\hat H$ in the standard SQD calculations proceeds in the same manner. 
Below, we present numerical results for $n=20$, $50$, and $100$.

Figures~\ref{fig:sqd-energy}(a)-\ref{fig:sqd-energy}(c) show the ground-state energy-estimation error 
\begin{equation}\label{eq:energy-err}
  \epsilon\coloneqq E_{\rm (F)SQD}-E_{\rm exact},  
\end{equation}
as a function of the number of shots $N_S$. 
Here, $E_{\rm (F)SQD}$ denotes the lowest eigenvalue of the truncated Hamiltonian matrix constructed in the sampled subspace, while $E_{\rm exact}$ is the essentially exact ground-state energy obtained from a high-accuracy DMRG calculation. 
We compare standard SQD using samples drawn from $\ket{\psi_g}$ with FSQD based on the projected filtered sampler states
$\hat P_{\ket{\bar{\bm0}}}\hat U_{\rm MPS,I}^{\dagger}\ket{\psi_g}$
and
$\hat P_{\ket{\bar{\bm0}}}\hat U_{\rm MPS,II}^{\dagger}\ket{\psi_g}$.
Figures~\ref{fig:sqd-energy}(d)--\ref{fig:sqd-energy}(f) show the corresponding sampled-subspace dimension $N_R$ as a function of $N_S$. 
A cutoff is imposed when $N_R$ exceeds the predefined threshold used in the numerical diagonalization of the truncated Hamiltonian matrix. 
Specifically, we set $N_R\leqslant5000$ for $n=20$ and $N_ R\leqslant8000$ for $n=50$ and $100$. 
%
A comparison with the dashed curves in Figs.~\ref{fig:sqd-energy}(d)--\ref{fig:sqd-energy}(f), which correspond to the unprojected filtered samplers, shows that the projected filtered samplers expand the sampled subspace more efficiently at fixed $N_S$. 
This directly reflects the role of the projection step in suppressing the dominant $\ket{\bm0}$ contribution and redistributing the sampling weight toward basis states useful for subspace construction.

\begin{figure}[bth]
    \includegraphics[width=\linewidth]{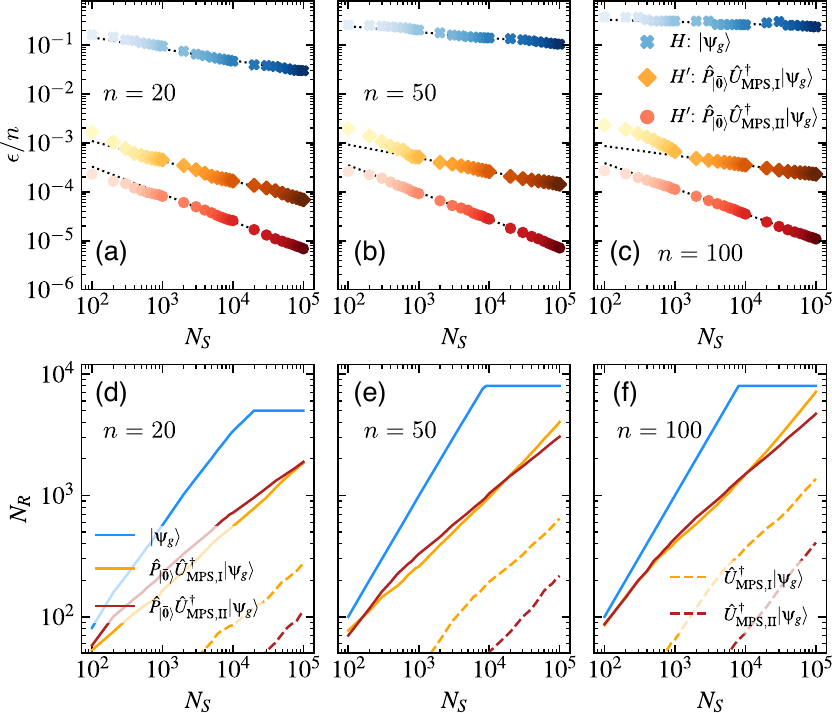}
    \caption{
    (a)--(c) Ground-state energy-estimation error $\epsilon/n$ versus the number of shots $N_S$ for standard SQD and FSQD. 
    The sampler state for SQD is the ground state $|\psi_g\rangle$ of the original Hamiltonian $\hat H$, wheras FSQD uses the projected filtered states $\hat P_{\ket{\bar{\bm0}}}\hat U_{\rm MPS,I}^{\dagger}\ket{\psi_g}$ and $\hat P_{\ket{\bar{\bm0}}}\hat U_{\rm MPS,II}^{\dagger}\ket{\psi_g}$. 
    The color intensity encodes $N_S$.
    Black dotted lines show fits of the form $\epsilon/n \propto N_S^{-\tau}$ using data in the regime $N_S\geqslant1000$. 
    The exact ground-state energies are $E_{\rm exact}/n = -1.285$, -1.300, and -1.305 for $n=20$, 50, and 100, respectively. 
    (d)--(f) Corresponding sampled-subspace dimension $N_R$ as a function of $N_S$, with cutoffs imposed at $N_R=5000$ for $n=20$ and $N_R=8000$ for $n=50$ and $100$. 
    For comparison, the results for unprojected filtered states $\hat U_{\rm MPS,I}^{\dagger}\ket{\psi_g}$ and $\hat U_{\rm MPS,II}^{\dagger}\ket{\psi_g}$ are also shown by dashed lines.
    }
    \label{fig:sqd-energy}
\end{figure}

\begin{table}[tbh]
\caption{\label{tab:sqd-rate}
Fitted decay exponent $\tau$ of the energy-estimation error, defined through $\epsilon/n \propto N_S^{-\tau}$. 
These exponents are obtained from numerical simulations for three different sampler states and from quantum-hardware experiments on \texttt{ibm\_kobe} for four experimentally relevant sampler constructions. 
The exponent $\tau$ quantifies how efficiently the estimation error decreases as the number of shots, or equivalently the number of circuit measurements in an experimental implementation, is increased.}
\begin{ruledtabular}
\begin{tabular}{lllll}
{Method} & {Sampler} & {$n=20$} & {$n=50$} & {$n=100$} \\
\hline
\multicolumn{5}{c}{simulation} \\
\hline
SQD & $\ket{\psi_g}$ & 0.235(7) & 0.120(5) & 0.04(2) \\
FSQD & $\hat P_{\ket{\bar{\bm0}}}\umpsdI\ket{\psi_g}$ & 0.40(2) & 0.274(4) & 0.195(2) \\
FSQD & $\hat P_{\ket{\bar{\bm0}}}\umpsdII\ket{\psi_g}$ & 0.559(5) & 0.553(8) & 0.518(4) \\
\hline
\multicolumn{5}{c}{experiment (\texttt{ibm\_kobe})} \\
\hline
SQD & $\ket*{\tilde\psi_g}$ & 0.328(3) & 0.149(7) & 0.002(2) \\
FSQD & $\umpsdI\ket*{\tilde\psi_g}$ & 0.68(3) & 0.51(3) & 0.18(2) \\
FSQD & $\hat U_{\ket{\bar{\bm0}}}\umpsdI\ket*{\tilde\psi_g}$ & 0.62(2) & 0.57(2) & 0.046(3) \\
FSQD & $\ket{\Psi_{\rm s}}$ & 0.60(3) & 0.61(4) & 0.12(2) \\
\end{tabular}
\end{ruledtabular}
\end{table}

Overall, FSQD yields energy-estimation errors that are several orders of magnitude smaller than those of the standard SQD method. 
The improvement is more pronounced for the two-layer filter than for the one-layer filter, consistent with the enhanced concentration of
$\hat U_{\rm MPS,II}^{\dagger}\ket{\psi_g}$ near $\ket{\bm0}$ as well as the reduced off-diagonal components of the filtered Hamiltonian matrix in the row and column associated with $|\bm 0\rangle$, as already seen in Fig.~\ref{fig:sqd-filter}(a) for the four-site system. 
This demonstrates that the performance gain originates from the MPS-circuit filter, which approximates the ground state with controlled accuracy and enhances the sparsity of the sampler while simultaneously making the filtered Hamiltonian more favorable for subspace construction around $\ket{\bm 0}$.

As $N_S$ increases, the sampled subspace expands and the estimation error decreases. 
To quantify this behavior, we fit the scaling relation
$\epsilon/n \propto N_S^{-\tau}$ in the regime $N_S\geqslant1000$, as indicated by the dotted lines in Figs.~\ref{fig:sqd-energy}(a)--\ref{fig:sqd-energy}(c). 
The fitted exponents are summarized in Table~\ref{tab:sqd-rate}. 
For the standard SQD, $\tau$ is small and decreases rapidly with increasing system size, indicating that adding more shots becomes increasingly inefficient for large systems. 
By contrast, FSQD exhibits substantially larger exponents, especially for the two-layer filter, and these exponents remain only weakly dependent on system size. 
Thus, the filter-assisted construction not only achieves much higher accuracy at a fixed number of shots, but also continues to benefit from further enlargement of the sampled subspace. 
For example, $\tau\simeq0.5$ implies that increasing $N_S$ by two orders of magnitude reduces the energy error by approximately one order of magnitude, demonstrating a clear sampling-overhead advantage over standard SQD. 


We also compare the observed errors with the theoretical estimates based on Theorem~\ref{theo:sqd}. 
As an example, for the $n=20$ case in Fig.~\ref{fig:sqd-energy}(a), the observed energy-estimation error at fixed $N_R$ is substantially smaller than the value suggested by the theorem. 
At $N_R=10^3$, the corresponding values of $N_S$ for 
$\hat U_{\rm MPS,I}^{\dagger}\ket{\psi_g}$ and
$\hat U_{\rm MPS,II}^{\dagger}\ket{\psi_g}$ 
are approximately $3\times10^{4}$ and $2\times10^{4}$, respectively, as shown in Fig.~\ref{fig:sqd-energy}(d). 
Using $\rho(\hat H)\propto n$, the corresponding rescaled errors $\tilde{\epsilon}\propto \epsilon/\rho(\hat H)$ are of order $10^{-4}$ and $2\times10^{-5}$.
By contrast, combined Theorem~\ref{theo:sqd} with the values of $\calL^{-1}$ extracted from Fig.~\ref{fig:sqd-sparsity}(b) at $\log_2N_R\approx10$ yields theoretical 
upper bounds of order $10^{-2}$ and $3\times10^{-3}$, respectively. 
Thus, although Theorem~\ref{theo:sqd} provides a rigorous upper-bound scale, the present benchmark indicates that the resulting upper bound is somewhat pessimistic 
when compared with the observed error.
Nevertheless, this discrepancy does not affect the scaling analysis of the resource requirements with system size, which is the central point of the theorem.

\begin{figure}[tb]
    \includegraphics[width=\linewidth]{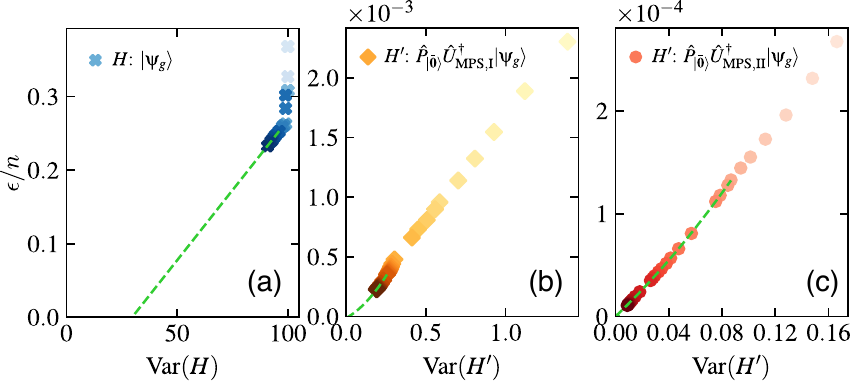}
    \caption{
    Ground-state energy-estimation error $\epsilon/n$ versus the energy variance $\mathrm{Var}(O) = \langle\nu_0|\hat O^2|\nu_0\rangle - \langle\nu_0|\hat O|\nu_0\rangle^2$ for $n=100$. 
    For standard SQD, $\hat O=\hat H$, whereas for FSQD, $\hat O=\hat H'$. 
    The color intensity corresponds to the number of shots $N_S$, matching the data shown in Fig.~\ref{fig:sqd-energy}(c). 
    Green dashed lines indicate a linear fit in (a) and quadratic fits in (b) and (c).
    }
    \label{fig:sqd-energyerr}
\end{figure}

A more accurate estimate of the ground-state energy can be obtained by energy-variance extrapolation~\cite{Sorella2001Generalized, Mizusaki2003Precise}, although the resulting estimate is no longer variational. 
For each sampled subspace, let $\ket{\nu_0}$ denote the lowest-energy eigenvector of the truncated Hamiltonian matrix. 
In our numerical simulations, we 
compute the variance
\begin{equation}
\mathrm{Var}(H)
=
\expval*{\hat H^2}{\nu_0}
-
\expval*{\hat H}{\nu_0}^2
\end{equation}
for standard SQD.
For FSQD, the same analysis is carried out using the filtered Hamiltonian $\hat H'$, yielding $\mathrm{Var}(H')$. 
Following the energy-variance analysis of Ref.~\cite{robledo2024chemistry_pub}, we plot the energy-estimation error $\epsilon$ against the corresponding variance for the $n=100$ system in Fig.~\ref{fig:sqd-energyerr}. 
The color intensity again represents the number of shots $N_S$, corresponding to those used in Fig.~\ref{fig:sqd-energy}(c).
Extrapolating the energy-variance relation to the zero-variance limit provides an estimate of the converged energy in the $N_S\to\infty$ limit, i.e., the exact ground-state energy.


For standard SQD, shown in Fig.~\ref{fig:sqd-energyerr}(a), the large-$N_S$ data exhibit an approximately linear dependence on the variance, although the variance itself remains large. 
The extrapolation gives $\epsilon_0/n=-0.113\pm0.007$, indicating that the sampled subspace is still far from the exact ground-state subspace and that the variance extrapolation does not yield a meaningful estimation. 
By contrast, the FSQD results in Figs.~\ref{fig:sqd-energyerr}(b) and \ref{fig:sqd-energyerr}(c), corresponding to the one- and two-layer filters, show much smaller variances and, simultaneously, much smaller ground-state energy-estimation errors as $N_S$ increases. 
Using a quadratic fit, we find that the zero-variance extrapolated offsets are
$\epsilon_0/n=(-1.28\pm1.30)\times10^{-5}$
and
$\epsilon_0/n=(0.56\pm1.99)\times10^{-7}$
for the one- and two-layer filters, respectively.
Thus, energy-variance extrapolation provides a further improvement beyond the already enhanced finite-sample accuracy of FSQD.


To summarize, the numerical benchmark demonstrates that FSQD substantially improves the sampling efficiency of subspace diagonalization for the quantum Ising model. 
The MPS-circuit filter strongly enhances the concentration of the ground-state wavefunction in the computational basis, thereby reducing the effective resource requirements for subspace construction. 
As a consequence, compared with the standard SQD protocol, FSQD achieves markedly smaller ground-state energy-estimation errors at a fixed number of shots, exhibits a much weaker system-size dependence of the sampling overhead, and retains its advantage as the system size increases. 
These numerical results are fully consistent with the Lorenz-curve and Gini-coefficient analyses presented in Sec.~\ref{sec:sparsity}, and confirm that the main advantage of FSQD lies in reshaping the sampling distribution itself rather than merely truncating the sampled subspace more aggressively.

Having established the intrinsic advantage of FSQD in the idealized noise-free setting, we now turn to its experimental realization on a quantum device. 
In the next section, we implement the corresponding SQD and FSQD protocols on the IBM quantum computer \texttt{ibm\_kobe} and examine to what extent the sampling-overhead advantage observed in the numerical benchmark persists under realistic hardware conditions. 
In particular, we compare different experimentally accessible sampler constructions and assess their impact on the accuracy and scaling of the resulting energy estimates.

\section{Experimental demonstration} 
\label{sec:expr}

\begin{figure*}[tb]
    \includegraphics[width=\linewidth]{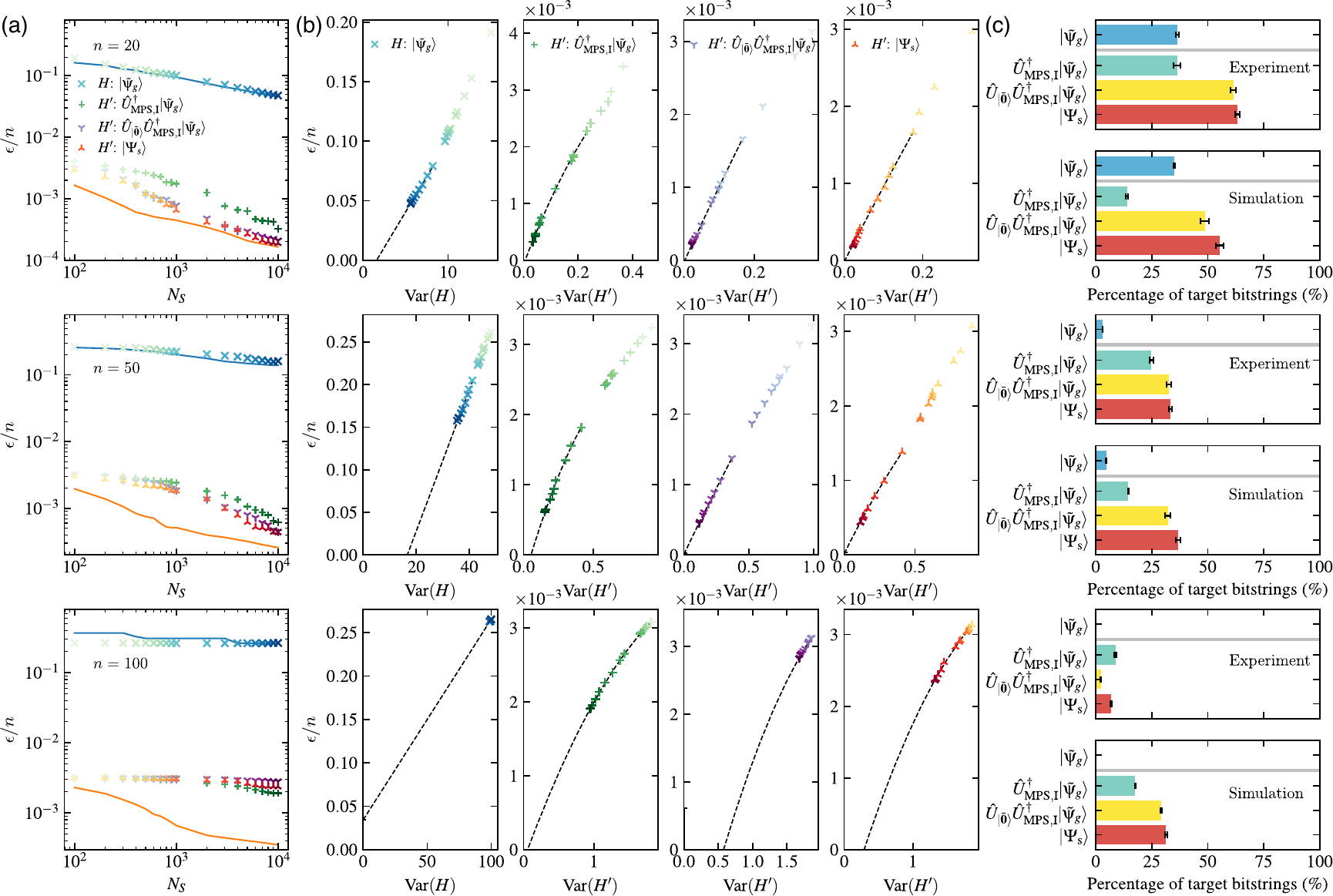}
    \caption{
    Experimental results for $n=20$ (top panels), $50$ (middle panels), and $100$ (bottom panels), using four different samplers $\ket*{\tilde\psi_g}$, $\umpsdI\ket*{\tilde\psi_g}$, $\hat U_{\ket{\bar{\bm0}}}\umpsdI\ket*{\tilde\psi_g}$, and $\ket{\Psi_{\rm s}}$. 
    (a) Ground-state energy-estimation error $\epsilon/n$ versus the number of shots $N_S$. 
    The color intensity corresponds to $N_S$, and the two solid curves reproduce the corresponding noise-free references from Figs.~\ref{fig:sqd-energy}(a)--\ref{fig:sqd-energy}(c), namely those for $|\psi_g\rangle$ (blue) and $\hat P_{\ket{\bar{\bm0}}}\hat U_{\rm MPS,I}^{\dagger}\ket{\psi_g}$ (red). 
    (b) Ground-state energy-estimation error $\epsilon/n$ versus the energy variance. 
    The black dashed lines denote linear fits for standard SQD and quadratic fits for the FSQD samplers. 
    (c) Percentage of target bitstrings obtained from hardware experiments and corresponding noiseless simulations. 
    }
    \label{fig:sqd-expr}
\end{figure*}

To assess the practical performance of FSQD on a real quantum computer, we implement both the standard SQD and FSQD protocols on IBM Quantum's Heron R2 processor \texttt{ibm\_kobe}~\cite{Qiskit}, applied to the quantum Ising model introduced in Sec.~\ref{sec:ising-model}. 
Additional experimental details, including the device qubit layout, coupling map, calibration data, and circuit compilation strategy, are provided in Appendix~\ref{appx:IV}.

Before presenting the experimental results, we summarize the four sampler states used to generate the quantum-selected computational basis states. 
First, for the standard SQD protocol, we prepare an approximate ground state of the original Hamiltonian $\hat H$, denoted by $\ket*{\tilde\psi_g}$, using the MPS circuit-encoding method described in Sec.~\ref{sec:automatic}. 
We use a circuit consisting of three layers of local two-qubit unitaries arranged in a brick-wall structure for $n=20$, and two such layers for $n=50$ and $n=100$.  

Second, as the simplest filter-assisted sampler, we apply the single-layer circuit filter to obtain
\begin{equation}
    \hat U_{\rm MPS,I}^{\dagger}\ket{\tilde\psi_g}.
\end{equation}
This state is experimentally accessible, but, as already discussed in Sec.~\ref{sec:fsqd}, it is theoretically expected to remain strongly concentrated on $\ket{\bm0}$ and therefore may not be efficient for subspace expansion.

Third, to mimic the projected sampler required in FSQD, we approximate the non-unitary projector $\hat P_{\ket{\bar{\bm0}}}$ by a unitary operator $\hat U_{\ket{\bar{\bm0}}}$ acting on the filtered state [see Eq.~(\ref{eq:U0})]. 
This yields the sequentially encoded projected sampler
\begin{equation}
    \hat U_{\ket{\bar{\bm0}}}\hat U_{\rm MPS,I}^{\dagger}\ket{\tilde\psi_g}.
\end{equation}
Here, $\hat U_{\ket{\bar{\bm0}}}$ is obtained by circuit encoding of the action of $\hat P_{\ket{\bar{\bm0}}}$ on the filtered state $\hat U_{\rm MPS,I}^{\dagger}\ket{\tilde\psi_g}$ and is implemented using a circuit consisting of two brick-wall layers of local two-qubit unitaries. 

Finally, as an alternative projected-state construction, we directly encode the normalized projected filtered state into a quantum circuit, yielding the sampler $\ket{\Psi_{\rm s}}$.
More specifically, $\ket{\Psi_{\rm s}}$ is obtained by direct circuit encoding of the normalized projected MPS corresponding to
\begin{equation}
    \frac{
    \hat P_{\ket{\bar{\bm0}}}\hat U_{\rm MPS,I}^{\dagger}\ket{\psi_g}
    }{
    \left\|
    \hat P_{\ket{\bar{\bm0}}}\hat U_{\rm MPS,I}^{\dagger}\ket{\psi_g}
    \right\|_2
    }.
\end{equation}
To mimic realistic classical preprocessing, we use an MPS approximation to $\ket{\psi_g}$ with bond dimension $\chi=6$ when constructing this target state.  
In this direct-encoding approach, we again use a circuit consisting of three brick-wall layers of local two-qubit unitaries for $n=20$, and two such layers for $n=50$ and $100$. 



Further details on the construction of $\hat U_{\ket{\bar{\bm0}}}$, the direct-encoding strategy, and the associated algorithmic errors are given in Appendix~\ref{appx:III}.
In summary, the four sampler states considered experimentally are $\ket*{\tilde\psi_g}$, $\umpsdI\ket*{\tilde\psi_g}$, $\hat U_{\ket{\bar{\bm0}}}\umpsdI\ket*{\tilde\psi_g}$, and $\ket{\Psi_{\rm s}}$. 
For each sampler, we measure all qubits in the computational basis for up to $N_S=10^4$ shots and use the resulting bitstrings to construct the sampled subspace.

Figure~\ref{fig:sqd-expr}(a) shows the ground-state energy-estimation error $\epsilon/n$ versus $N_S$, in analogy with the numerical results shown in Figs.~\ref{fig:sqd-energy}(a)--\ref{fig:sqd-energy}(c). 
The colored markers denote the experimental results obtained directly from raw bitstring samples. 
The solid curves show the corresponding noise-free numerical references based on the ideal sampler states $|\psi_g\rangle$ and $\hat P_{\ket{\bar{\bm0}}}\hat U_{\rm MPS,I}^{\dagger}\ket{\psi_g}$, i.e., the same numerical results already presented in Figs.~\ref{fig:sqd-energy}(a)--\ref{fig:sqd-energy}(c). 
Consistent with the numerical simulations, the filter-assisted samplers generally achieve smaller ground-state energy-estimation errors than the standard SQD sampler for all system sizes considered. 
Moreover, the fitted decay exponents $\tau$, listed in Table~\ref{tab:sqd-rate}, show that the filter-assisted samplers typically exhibit faster error reduction than the standard SQD sampler as $N_S$ increases. 

Among the three filter-assisted samplers, the two projected-state constructions, namely $\hat U_{\ket{\bar{\bm0}}}\hat U_{\rm MPS,I}^{\dagger}\ket{\tilde\psi_g}$ and $\ket{\Psi_{\rm s}}$, yield noticeably smaller errors than the directly filtered state $\hat U_{\rm MPS,I}^{\dagger}\ket{\tilde\psi_g}$ for $n=20$ and $50$, indicating that suppressing the $\ket{\bm0}$ component remains beneficial for subspace expansion even on hardware. 
For $n=100$, however, this tendency becomes less clear, and the directly filtered state $\hat U_{\rm MPS,I}^{\dagger}\ket{\tilde\psi_g}$ can yield a better energy estimate in the experimentally accessible range of $N_S$, especially for large $N_S$. 
This behavior is plausibly attributed to hardware noise, which broadens the sampling distribution and alters the balance between concentration on $\ket{\bm0}$ and efficient subspace expansion. 
On the other hand, the experimentally fitted decay exponents $\tau$ remain of similar magnitude among the three filtered samplers, suggesting that their large-$N_S$ scaling behavior is less clearly distinguishable on present hardware than in the ideal numerical benchmark. 
Figure~\ref{fig:sqd-expr}(b) presents the corresponding energy-variance relations and extrapolations to zero variance limit, providing a complementary view of the experimental performance..



To investigate the impact of hardware noise more directly, we define the target bitstrings $\{x_k\}_{k=1}^{N_R}$ as those obtained from the corresponding ideal sampler at a given number of shots $N_S$. 
More specifically, for standard SQD the target set is generated from the exact ground state $\ket{\psi_g}$, 
whereas for the filter-assisted samplers the target set is generated from the ideal projected state $\ket{\bar{\bm0}}_g$ defined in Eq.~(\ref{eq:0_bar}) with $\hat U_{\rm MPS}=\hat U_{\rm MPS,I}$. 
For each sampler state, we obtain a set of sampled bitstrings $\{\tilde x_k\}$ from $N_S$ shots and compute the fraction $\theta \coloneqq \bar N_R/N_R$ of target bitstrings recovered in the sampled set, where $\{\bar x_k\}_{k=1}^{\bar N_R} = \{x_k\}\cap\{\tilde x_k\}$. 
For this analysis, we fix $N_S=10^4$ and use $M_c=5$ independent runs to compute the sample mean and standard deviation,
\begin{equation}
    \bar{\theta}
    =
    \frac{1}{M_c}\sum_{c=1}^{M_c}\theta_c,
    \quad
    \epsilon_\theta
    =
    \sqrt{
    \frac{1}{(M_c-1)}
    \sum_{c=1}^{M_c}\bigl(\theta_c-\bar{\theta}\bigr)^2
    }.
\end{equation}
These are indicated by the rectangular bars and error bars, respectively, in Fig.~\ref{fig:sqd-expr}(c).

The fraction of target bitstrings provides an interpretable measure of sampler quality.
For example, although the MPS-based encoding of the ground state is highly effective, the target-bitstring fraction for the ground-state sampler becomes very small at large system size, reflecting the low sparsity of the ground state itself.
For the filter-assisted samplers, the simulated results decrease with increasing system size due to the error arising from imperfect circuit encoding of the intended target sampler, whereas in experiment this percentage decreases even more rapidly, indicating significant hardware noise.
At $n=20$ and $50$, the experiment shows an increase in the target-bitstring percentage, suggesting that noise broadens the sampling distribution and improves subspace expansion. This effect is especially visible for the directly filtered state $\umpsdI\ket*{\tilde\psi_g}$ and explains the modest reduction of its energy-estimation error. Such a noise-assisted effect disappears for $n=100$.

\section{Summary and Outlook}\label{sec:summary}

In this work, we introduced filter-assisted subspace quantum diagonalization (FSQD), a quantum-centric framework for improving sample-based subspace diagonalization by engineering the sampling distribution itself. 
The central idea is to apply a quantum filter, i.e., a unitary transformation of the Hamiltonian, such that the ground-state wavefunction of the transformed problem becomes substantially more sparse in the computational basis. 
This directly mitigates the intrinsic trade-off between sampling efficiency and wavefunction sparsity that limits standard SQD and related sample-based approaches.

Using the Lorenz curve and the Gini coefficient as quantitative measures of wavefunction sparsity, we established a direct connection between the concentration of the measurement distribution and the resource requirements of SQD. 
This allowed us to derive rigorous bounds on the sampled-subspace dimension and the number of measurement shots required to achieve a target accuracy. 
These results clarify why standard SQD can suffer from exponentially growing sampling overhead in strongly correlated systems and, conversely, why enhancing the sparsity of the sampling distribution is essential for scalable performance.

To realize the quantum filter, we employed a tensor-network-based circuit-encoding method that maps an approximate target state to a shallow quantum circuit composed of local two-qubit unitaries. 
This construction provides a practical route for implementing FSQD both in classical numerical simulations and on gate-based quantum hardware. 
Importantly, the method also admits projected-sampler constructions that suppress the dominant $\ket{\bm 0}$ component and thereby restore efficient subspace expansion.

We benchmarked FSQD for the quantum Ising model with transverse and longitudinal fields using both ideal numerical simulations and quantum-hardware experiments. 
In the numerical benchmark, FSQD substantially enhanced wavefunction sparsity, reduced the system-size dependence of the effective sampling overhead, and achieved ground-state energy-estimation errors that were orders of magnitude smaller than those of standard SQD at a comparable number of samples. 
Energy-variance extrapolation provided a further improvement, yielding highly accurate energy estimates for the filtered samplers. 
These results were fully consistent with the Lorenz-curve and Gini-coefficient analyses, confirming that the main advantage of FSQD lies in reshaping the sampling distribution rather than merely truncating the sampled subspace more aggressively.

We further demonstrated the protocol experimentally on IBM Quantum's \texttt{ibm\_kobe} processor using several practically realizable sampler constructions. 
The hardware results show that the advantage of FSQD persists on present-day quantum devices, especially for the 20-qubit case, where the projected-state samplers clearly outperform both standard SQD and the unprojected filtered sampler. 
Although FSQD continues to outperform standard SQD at larger system sizes, hardware noise and imperfect sampler preparation become the dominant limitations, making the differences among the filtered samplers less pronounced. 
These observations highlight both the promise of FSQD on near-term quantum hardware and the need for improved state preparation, better projector approximations, and classical post-processing techniques such as configuration recovery in order to extend its practical reach.

Finally, our results establish FSQD as a promising route toward scalable sample-based quantum many-body calculations. 
Because the protocol is formulated at the level of sampling distributions and subspace construction, it is not restricted to the quantum Ising model studied here and should be applicable more broadly to strongly correlated lattice models and quantum chemistry problems. 
An important direction for future work is to combine FSQD with improved sampler design, symmetry adaptation, and error-mitigation or post-processing strategies, thereby further extending the practical reach of quantum-selected subspace methods on larger and noisier quantum devices. 


\section*{Acknowledgments}
We are grateful to Xiaoyang Wang for valuable discussions. 
Part of the numerical simulations was performed using the HOKUSAI supercomputer at RIKEN (Project ID RB240003). 
A portion of this work is based on results obtained from Project No.~JPNP20017, supported by the New Energy and Industrial Technology Development Organization (NEDO). 
This study was also supported by JSPS KAKENHI Grants No.~JP21H04446 and No.~JP26K06972. 
Additional support was provided by JST COI-NEXT (Grant No.~JPMJPF2221) and the Program for Promoting Research on the Supercomputer Fugaku (Grant No.~MXP1020230411) by MEXT, Japan. 
We also acknowledge support from the UTokyo Quantum Initiative, the RIKEN TRIP initiative (RIKEN Quantum), and the COE research grant in computational science from Hyogo Prefecture and Kobe City through the Foundation for Computational Science.

\section*{Data Availability}
The data that support the findings of this study, as well as the code used to generate the figures in this paper, are available from the corresponding author upon reasonable request. 

\appendix

\section{Monotonic improvement of the SQD ground-state energy estimate}\label{appx:I}

In this appendix, we show that the ground-state energy estimate obtained from SQD is monotonically non-increasing as the sampled-subspace dimension $N_R$ is increased. 
More precisely, when the sampled subspace is enlarged by adding one additional computational basis state, the lowest eigenvalue of the corresponding truncated Hamiltonian cannot increase.

Let $H_{N_R}$ denote the truncated Hamiltonian matrix in a sampled subspace of dimension $N_R$, and let its eigenvalues be ordered as
\begin{equation}
E_0 \leqslant E_1 \leqslant \dots \leqslant E_{N_R-1}.
\end{equation}
Now enlarge the sampled subspace by adding one additional computational basis state.
The corresponding truncated Hamiltonian matrix in the enlarged $(N_R+1)$-dimensional subspace can be written as
\begin{equation}
H_{N_R+1}=
\begin{pmatrix}
H_{N_R} & H\vec e_{N_R+1} \\
\vec e_{N_R+1}^{\;T}H & a_1
\end{pmatrix},
\end{equation}
where $\vec e_{N_R+1}$ denotes the basis vector associated with the added bitstring, and $a_1=\vec e_{N_R+1}^{\;T}H\vec e_{N_R+1}$. 
Here $H$ denotes the matrix representation of the full Hamiltonian $\hat H$ in the computational basis.

Let $\ket{\nu_0}$ be the normalized eigenstate of $H_{N_R}$ with the lowest eigenvalue $E_0$, and let $\vec{\nu}_0$ denote its vector representation in the computational basis. 
Following the Lanczos construction~\cite{koch2011}, we define
\begin{equation}
a_0=\expval*{\hat H}{\nu_0}
=\vec{\nu}_0^{\;T}H_{N_R}\vec{\nu}_0
=E_0,
\end{equation}
and the next normalized basis vector $\ket{\nu_1}$ by
\begin{equation}
b_1\ket{\nu_1}=\hat H\ket{\nu_0}-a_0\ket{\nu_0},
\end{equation}
where
\begin{equation}
b_1=\mel{\nu_1}{\hat H}{\nu_0}.
\end{equation}
Restricting $\hat H$ to the two-dimensional Krylov subspace
\begin{equation}
\mathrm{span}(\ket{\nu_0},\hat H\ket{\nu_0}),
\end{equation}
we obtain the corresponding Hamiltonian matrix 
\begin{equation}
H_{\mathrm{span}(\ket{\nu_0},\hat H\ket{\nu_0})}
=
\begin{pmatrix}
a_0 & b_1 \\
b_1 & a_1
\end{pmatrix}.
\end{equation}
Its lowest eigenvalue is
\begin{equation}
\check E_0
=
\frac{1}{2}
\left[
a_0+a_1-\sqrt{(a_0-a_1)^2+4b_1^2}
\right].
\end{equation}
Since
\begin{equation}
\sqrt{(a_0-a_1)^2+4b_1^2}\ge |a_0-a_1|,
\end{equation}
we have
\begin{equation}
\check E_0
\leqslant
\frac{1}{2}(a_0+a_1-|a_0-a_1|)
=
\min(a_0,a_1)
\leqslant a_0=E_0.
\end{equation}
Therefore, the lowest energy in the enlarged subspace cannot exceed the lowest energy in the original sampled subspace.
This proves that the SQD ground-state energy estimate is monotonically non-increasing as $N_R$ increases.

Furthermore, when $a_1\geqslant a_0$, a Taylor expansion of $\check E_0$ in the small parameter $b_1$ gives
\begin{equation}
\check E_0-E_0
=
-\frac{b_1^2}{|a_0-a_1|}
+O(b_1^4).
\end{equation}
This shows that the rate of improvement depends on the coupling $b_1$ between the current approximate ground state and the newly added basis direction, as well as on the local energy scale $|a_0-a_1|$. 
In other words, the convergence speed is governed by detailed properties of the Hamiltonian matrix and by the overlap structure of the sampled subspace.


\section{Bounds on SQD Resource Requirements from the Lorenz Curve and Gini Coefficient}\label{appx:II}

In this appendix, we derive the relationship between wavefunction sparsity and the computational resources required by the standard SQD method, following an argument closely related to that of Ref.~\cite{robledo2024chemistry_pub}. As discussed in Sec.~\ref{sec:sparsity}, we use the Lorenz curve and the Gini coefficient as quantitative measures of concentration of the measurement distribution.

We begin by deriving a lower bound on the number of measurement shots required to sample the relevant bitstrings.
From the graphical interpretation of the Gini coefficient, we regard the first $GN$ bitstrings as the least prominent configurations and the remaining $(1-G)N$ bitstrings as the prominent ones.
Because the probabilities are sorted in non-decreasing order, every prominent configuration satisfies 
\begin{equation}
|c_I|^2 \geqslant \frac{\mathcal{L}'(G)}{N},
\qquad \text{for } I\geqslant GN,
\end{equation}
where $\calL'(G)$ denotes the slope of the Lorenz curve at the normalized index $x=G$. 
Since $\mathcal{L}(x)$ is convex and satisfies $\mathcal{L}(1)=1$, the tangent line at $x=G$ lies below the point $(1,1)$, implying
\[
\mathcal{L}'(G)\le \frac{1-\mathcal{L}(G)}{1-G}\le \frac{1}{1-G},
\]
where we used $\mathcal{L}(G)\ge 0$. 

If $N_S$ shots are performed, the probability that not all prominent configurations have been observed is upper bounded by
\begin{equation}\label{eq:bound_of_failure}
\begin{aligned}
p_{\rm fail}
&=
\sum_{I=GN+1}^{N}(1-|c_I|^2)^{N_S} \\
&\leqslant
\sum_{I=GN+1}^{N}\left(1-\frac{\mathcal{L}'(G)}{N}\right)^{N_S} \\
&\leqslant
(1-G)N\,e^{-N_S\mathcal{L}'(G)/N} \\
&\leqslant
\frac{N}{\mathcal{L}'(G)}\,e^{-N_S\mathcal{L}'(G)/N}.
\end{aligned}
\end{equation}
In deriving the third line, we used the elementary bound $1-a\le e^{-a}$ for real $a$. 
If we require the failure probability to be at most $\eta$ with $\eta>0$, then the right-hand side of Eq.~\eqref{eq:bound_of_failure} must satisfy $\leqslant \eta$. 
Hence, a sufficient number of shots is 
\begin{equation}\label{eq:bound_of_sample}
N_S \geqslant
\frac{\log\!\bigl(N/(\mathcal{L}'(G)\eta)\bigr)}{\mathcal{L}'(G)/N}
\geqslant
\frac{\log\!\bigl[(1-G)N/\eta\bigr]}{\mathcal{L}'(G)/N},
\end{equation}
where the first inequality follows from the last line of Eq.~\eqref{eq:bound_of_failure}, and the second follows from the third line of Eq.~\eqref{eq:bound_of_failure}. 
The same reasoning remains valid after replacing $G$ by an arbitrary normalized index $x\in[0,1]$. 
Note that for notational simplicity, we treat quantities such as $GN$ and $xN$ as integers; otherwise one may replace them by the corresponding floor or ceiling functions without affecting the argument.

We next prove Theorem~\ref{theo:sqd} by bounding the SQD deviation
\begin{equation}
\Delta_{\mathrm{SQD}}
\coloneqq
E_{\mathrm{SQD}}-\expval{\hat H}{\psi}.
\end{equation}
For a given sampled subspace, define the normalized truncated state 
\begin{equation}
    \ket{\psi_{ x }} \coloneqq \frac{1}{\sqrt{1-\calL(x)}}\sum_{I=xN+1}^{N}c_I\ket{x_I}, 
\end{equation}
where $x\in[0,1]$ denotes the normalized cutoff index. Here, we implicitly assume that the sampled subspace consists of the $(1-x)N$ most prominent bitstrings.  
The Euclidean distance between $\ket{\psi}$ and $\ket{\psi_{ x }}$ is 
\begin{equation}\label{eq:error-distance}
    \begin{aligned}
        \Vert \ket{\psi}\pm\ket{\psi_{ x }}\Vert_2^2 &= 2 \pm 2\frac{1}{\sqrt{1-\calL(x)}}\sum_{I=xN+1}^{N}|c_I|^2 \\   
        &= 2 \pm 2\sqrt{1-\calL(x)}.
    \end{aligned}
\end{equation} 
Let $H_{N_R}$ denote the truncated Hamiltonian matrix in the sampled subspace, and let $\ket{\nu_0}$ be its lowest eigenstate. By construction of the sampled subspace and because $\ket{\nu_0}$ minimizes the energy within it, we have
\begin{equation}\label{eq:error_expval}
    \begin{aligned}
        &\expval{\hat H}{\nu_0} = \expval{H_{N_R}}{\nu_0} \\
        & \leqslant \expval{H_{N_R}}{\psi_{ x }} = \expval{\hat H}{\psi_{ x }}.
    \end{aligned}
\end{equation}

Now define $\ket{\psi'}\coloneqq\ket{\psi_x}-\ket{\psi}$. Then 
\begin{equation}\label{eq:use_psi_prime}
    \expval{\hat H}{\psi_{ x }} - \expval{\hat H}{\psi} = \mel{\psi'}{\hat H}{\psi_{ x }} + \mel{\psi}{\hat H}{\psi'}. 
\end{equation}
Taking the absolute value and applying the triangle inequality yields 
$|\expval{\hat H}{\psi_{ x }} - \expval{\hat H}{\psi}|   \leqslant |\mel{\psi'}{\hat H}{\psi_{ x }}| + |\mel{\psi}{\hat H}{\psi'}|$. 
Using the Cauchy--Schwarz inequality, 
$|\braket*{\psi}{\hat H\ket{\psi'}}| \leqslant \Vert\psi\Vert_2 \Vert\hat H\psi'\Vert_2=\Vert\hat H\psi'\Vert_2$ and $|\braket*{\psi_{ x }}{\hat H\ket{\psi'}}| \leqslant \Vert\psi_{ x }\Vert_2 \Vert\hat H\psi'\Vert_2=\Vert\hat H\psi'\Vert_2$, we obtain the inequality 
\begin{equation}
    |\expval{\hat H}{\psi_{ x }} - \expval{\hat H}{\psi}| \leqslant 2 \Vert\hat H\psi'\Vert_2 \leqslant 2 \rho(\hat H)\Vert\psi'\Vert_2,
\end{equation}
where $\rho(\hat H)$ denotes the spectral norm of the Hermitian operator $\hat H$, i.e., the largest singular value of the matrix $H$. Since $H$ is Hermitian, this is equal to the maximum absolute value of its eigenvalues $\lambda_1,\dots,\lambda_N$.

Combining this with Eqs.~\eqref{eq:error-distance} and \eqref{eq:error_expval}, we obtain the upper bound on $\Delta_{\mathrm{SQD}}$: 
\begin{equation}\label{eq:error-bound}
    \expval{\hat H}{\nu_0} - \expval{\hat H}{\psi} \leqslant 2\sqrt{2}\rho(\hat H)\big(1-\sqrt{1-\calL(x)}\big)^{1/2}
\end{equation}
If this upper bound is required to be at most $\varepsilon$, then Eq.~\eqref{eq:error-bound} implies 
\begin{equation}
    \calL(x) \leqslant 1-(1-\tilde\varepsilon^2)^2 = 2\tilde\varepsilon^2 - \tilde\varepsilon^4,
\end{equation}
where $\tilde{\varepsilon}=\varepsilon/(2\sqrt{2}\rho(\hat H))$.
Assuming $\tilde{\varepsilon}\leqslant 1$, we may use the simpler sufficient condition
\begin{equation}
\mathcal{L}(x)\leqslant \tilde\varepsilon^2,
\end{equation}
because $\tilde{\varepsilon}^2 \leqslant 2\tilde{\varepsilon}^2-\tilde{\varepsilon}^4$, which gives 
\begin{equation}
    1-x \geqslant 1-\calL^{-1}(\tilde\varepsilon^2).
\end{equation}
Therefore, a sufficient sampled-subspace dimension is
\begin{equation}
N_R=(1-\mathcal{L}^{-1}(\tilde\varepsilon^2))N.
\end{equation}
By Eq.~\eqref{eq:bound_of_sample}, the corresponding sufficient number of shots is
\begin{equation}\label{eq:bound_of_shot_fullver}
    N_S = \frac{\log[(1-\calL^{-1}(\tilde\varepsilon^2))N/\eta]}{\calL'(\calL^{-1}(\tilde\varepsilon^2))/N}
\end{equation}
which is sufficient to sample these $N_R$ bitstrings with probability at least $1-\eta$, thereby proving Theorem~\ref{theo:sqd}.

In fact, the upper bound in Eq.~\eqref{eq:error-bound} can be improved when the relevant states have real amplitudes in the computational basis. 
We introduce $\ket{\psi''}\coloneqq\ket{\psi}+\ket{\psi_{ x }}$ and, instead of Eq.~\eqref{eq:use_psi_prime}, use
\begin{equation}
    \expval{\hat H}{\psi_{ x }} - \expval{\hat H}{\psi} =  \mel{\psi'}{\hat H}{\psi''},
\end{equation}
where $\Im[\mel{\psi}{\hat H}{\psi_x}]=0$ has been used. 
Applying the Cauchy--Schwarz inequality together with Eq.~\eqref{eq:error-distance}, we obtain the improved bound 
\begin{equation}
    \expval{\hat H}{\nu_0} - \expval{\hat H}{\psi} \leqslant 2\rho(\hat H)\sqrt{\calL(x)}.
\end{equation}
Accordingly, one may redefine $\tilde{\varepsilon}=\frac{\varepsilon}{2\rho(\hat H)}$, and the subsequent argument proceeds analogously.

Now, we estimate how the sufficient subspace dimension and number of shots scale with the system size when an analytical form of the Lorenz curve is available. 
As a first example, suppose that the configuration weights decay exponentially as 
\begin{equation}
|c_I|^2 \propto e^{-\lambda(N-I)},
\qquad \lambda N\gg 1,
\end{equation}
where $\lambda$ controls the concentration of the distribution, with larger $\lambda$ corresponding to stronger sparsity. 
In this regime, the Lorenz curve is approximately given by
\begin{equation}
\mathcal{L}(x)\simeq e^{-\lambda(1-x)N},
\end{equation}
with derivative
\begin{equation}
\mathcal{L}'(x)\simeq \lambda N e^{-\lambda(1-x)N},
\end{equation}
and inverse function
\begin{equation}
\mathcal{L}^{-1}(y)\simeq 1-\frac{\log(1/y)}{\lambda N},
\end{equation}
all as functions of the normalized index $x=I/N$.
Substituting these expressions into Theorem~\ref{theo:sqd}, we obtain 
\begin{gather}
        N_R = \frac{\log(1/\tilde\varepsilon^{2})}{\lambda}, 
        \label{eq:N_R-exp-decay} \\
        N_S = \frac{1}{\lambda\tilde\varepsilon^{2}}\log(\frac{\log(1/\tilde\varepsilon^{2})}{\lambda\eta})
            = \frac{1}{\lambda\tilde\varepsilon^{2}}\log(\frac{N_R}{\eta}). 
            \label{eq:N_S-exp-decay}
\end{gather} 
Consequently, the Gini coefficient can be evaluated as
\begin{equation}
1-G
=
2\int_0^1 \mathcal{L}(x)\,dx
=
\frac{2(1-e^{-\lambda N})}{\lambda N},
\end{equation}
which reduces to
\begin{equation}
1-G \simeq \frac{2}{\lambda N}
\end{equation}
in the regime $\lambda N\gg 1$.

This expression reflects the role of the decay rate $\lambda$:
Smaller $\lambda$ means slower decay of the configuration weights, a more balanced distribution (i.e. less sparse), and hence a smaller Gini coefficient.
In particular, when $\lambda$ is independent of system size, the Gini coefficient exhibits the same scaling as the sparsest-case result in Eq.~\eqref{eq:most-sparse}. 
Using the asymptotic relation $1-G\simeq 2/(\lambda N)$, Eqs.~\eqref{eq:N_R-exp-decay} and \eqref{eq:N_S-exp-decay} can be rewritten as
\begin{gather}
        N_R = \frac{(1-G)N}{2}\log(1/\tilde\varepsilon^{2}), \\
        N_S = \frac{(1-G)N}{2\tilde\varepsilon^{2}}\log(\frac{N_R}{\eta}).
\end{gather}

As a second example, consider the power-law decay 
\begin{equation}
|c_I|^2 \propto \frac{1}{(N-I+\lambda^{-1})^\gamma},
\qquad \lambda>0,\ \gamma>1,\ \lambda N\gg 1.
\end{equation}
As in the exponential case, the rate parameter $\lambda$ controls the sparsity of the distribution, with larger $\lambda$ corresponding to stronger concentration on the most prominent configurations. 
The Lorenz curve, its derivative, and its inverse are then
\begin{equation}
\mathcal{L}(x)\simeq(\lambda N)^{1-\gamma}\left(1+\frac{1}{\lambda N}-x\right)^{1-\gamma},
\end{equation}
\begin{equation}
\mathcal{L}'(x)\simeq(\gamma-1)(\lambda N)^{1-\gamma}\left(1+\frac{1}{\lambda N}-x\right)^{-\gamma},
\end{equation}
and
\begin{equation}
\mathcal{L}^{-1}(y)\simeq1+\frac{1}{\lambda N}-\frac{1}{\lambda N}y^{\frac{1}{1-\gamma}},
\end{equation}
respectively.
Combining these expressions with Theorem~\ref{theo:sqd}, we obtain 
\begin{gather}
    N_R = \frac{\tilde\varepsilon^{\frac{2}{1-\gamma}}-1}{\lambda}, \label{eq:N_R-power-decay} \\
    N_S = \frac{1}{(\gamma-1)\lambda \tilde\varepsilon^{\frac{2\gamma}{\gamma-1}}} \log(\frac{N_R}{\eta}). \label{eq:N_S-power-decay}
\end{gather}

In this case, the Gini coefficient satisfies 
\begin{equation}
1-G=\frac{2}{\lambda N}\log(1+\lambda N)
\qquad \text{for } \gamma=2,
\end{equation}
and
\begin{equation}
1-G=\frac{2}{(\gamma-2)\lambda N}
+O\!\left(\frac{1}{(\lambda N)^{\gamma-1}}\right)
\qquad \text{for } \gamma>2,
\end{equation}
which is surprisingly similar to the exponential-decay case, indicating the usefulness of the Gini coefficient as a measure of wavefunction sparsity. 
Substituting the Gini coefficient for the rate parameter and ignoring higher-order corrections, we obtain from Eqs.~\eqref{eq:N_R-power-decay} and \eqref{eq:N_S-power-decay} that, for $\gamma>2$, 
\begin{gather}
    N_R = \frac{(1-G)N(\gamma-2)}{2}(\tilde\varepsilon^{\frac{2}{1-\gamma}}-1),\\
    N_S = \frac{(1-G)N(\gamma-2)}{2(\gamma-1) \tilde\varepsilon^{\frac{2\gamma}{\gamma-1}}} \log(\frac{N_R}{\eta}).
\end{gather}

These analytical examples clarify that the sampling overhead of standard SQD is controlled by the system-size dependence of the Lorenz curve and, in particular, by the scaling of the Gini coefficient.
When $(1-G)N$ grows rapidly with system size, both the sufficient sampled-subspace dimension $N_R$ and the required number of shots $N_S$ become large, leading to severe sampling overhead.
The role of FSQD is to modify this scaling behavior at the level of the measurement distribution itself.
By applying a quantum filter, FSQD aims to transform the original sampler into a filtered one whose probability distribution is substantially sparser in the computational basis, thereby reducing $(1-G)N$ and weakening its system-size dependence (see Fig.~\ref{fig:sqd-sparsity} for a numerical demonstration).
Ideally, this corresponds to a smaller effective exponent $g$ in the scaling relation $(1-G)N\propto 2^{gn+c}$, or even to a system-size-independent Gini coefficient.
In this way, FSQD can mitigate the sampling overhead of standard SQD not merely by reducing the sampled-subspace dimension at fixed distribution, but by reshaping the distribution itself so that the basis states relevant for subspace construction are sampled much more efficiently.

\section{Performance of the MPS-based circuit encoding method} \label{appx:III}

In this appendix, we assess the performance of the MPS-based circuit-encoding method used in Secs.~\ref{sec:ising-model} and \ref{sec:expr}. 
Our purpose is to quantify how accurately this method can encode (i) the ground-state filter circuit, (ii) the unitary approximation to the projector used in the sequential projected-sampler construction, and (iii) the directly encoded projected sampler $\ket{\Psi_{\rm s}}$. 
Unless otherwise stated, the results shown in this appendix are obtained for the quantum Ising model defined in Eq.~\eqref{eq:ising-model}.

To align the notation with the main text, we distinguish three levels of description of the ground state. 
The state $\ket{\psi_g}$ denotes the exact ground state of $\hat H$. 
The state $\ket{\psi_b(\chi)}$ denotes an MPS approximation to $\ket{\psi_g}$ with finite bond dimension $\chi$, and hence describes the classical approximation error introduced by truncating the bond dimension. 
Finally, $\ket{\tilde\psi_g}$ denotes the circuit-encoded approximation to $\ket{\psi_g}$, which is also used in the hardware experiments. 
Accordingly, the ideal filtered state is
$\ket{\psi_g'}$, defined in Eq.~(\ref{eq:psip}), 
while the ideal projected state $\ket{\bar{\bm0}}_g$ is defined in Eq.~\eqref{eq:0_bar}.

\begin{figure}[tbh]
    \includegraphics[width=\linewidth]{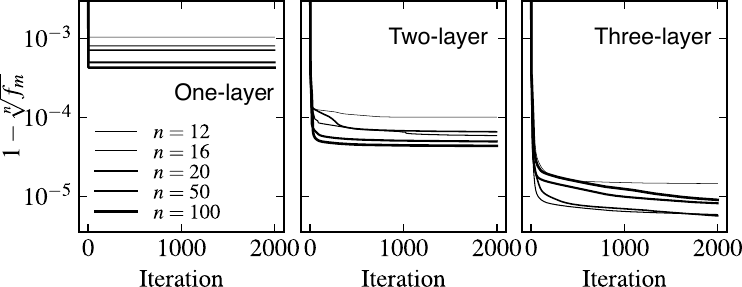}
    \caption{
      Optimization trajectories of the MPS-based circuit-encoding algorithm for system sizes $n=12,16,20,50,100$.
      From left to right, the encoded circuits consist of one, two and three brick-wall layers of local two-qubit unitaries. 
      One iteration corresponds to a single forward sweep followed by a single backward sweep. 
      }
    \label{fig:mps-trajectory}
\end{figure}

We first examine the quality of the circuit encoding of the ground state itself, namely the approximation
\begin{equation}
\ket{\psi_g}\approx \hat U_{\rm MPS}\ket{\bm0}, 
\end{equation}
where $\hat U_{\rm MPS}$ is the circuit obtained by the MPS-circuit-encoding method described in Sec.~\ref{sec:automatic}. 
Note that $\hat U_{\rm MPS}$ is used both as a circuit filter and, when applied to $\ket{\bm0}$, as a circuit-based approximation to the ground state $|\tilde\psi_g\rangle$. 
Figure~\ref{fig:mps-trajectory} shows the corresponding optimization trajectories for system sizes $n=12,16,20,50,100$ and for circuits composed of one, two, and three brick-wall layers of local two-qubit unitaries. 
As a performance measure, we plot the infidelity per site, defined by $1-\sqrt[n]{f_m}$, as a function of the iteration step. 
Here, one iteration step corresponds to a forward sweep followed by a backward sweep in the MPS-circuit-encoding algorithm (see Fig.~\ref{fig:tn-alg} and Algorithm~\ref{alg:two}). 
The objective function $f_m$ is defined in Eq.~(\ref{eq:mps-objective-func}); in the present context, it measures the overlap between the target state and the encoded state, i.e., $f_m=|\langle\psi_g|\hat U_{\rm MPS}\ket{\bm0}|$.

As shown in Fig.~\ref{fig:mps-trajectory}, a single iteration step is sufficient to achieve convergence for the one-layer circuit encoding. For the two-layer encoding, approximately 2000 iteration steps are required to obtain converged circuits, but the resulting infidelity per site is an order of magnitude smaller than in the one-layer case. 
For the three-layer encoding, the optimization has not yet fully converged after 2000 iteration steps, yet the infidelity per site is already nearly one order of magnitude smaller than in the two-layer case. 
The one- and two-layer circuit filters $\hat U_{\rm MPS,I}$ and $\hat U_{\rm MPS,II}$ used in the numerical simulations of Sec.~\ref{sec:ising-model} and the quantum-hardware experiments of Sec.~\ref{sec:expr} are taken directly from the corresponding one- and two-layer circuit encodings obtained here. 
By contrast, the circuit-encoded ground state $|\tilde\psi_g\rangle$ used in the hardware experiments is prepared using the three-layer encoding for $n=20$ and the two-layer encoding for $n=50$ and $n=100$,  also obtained here. 
\begin{figure}[tbh]
    \includegraphics[width=\linewidth]{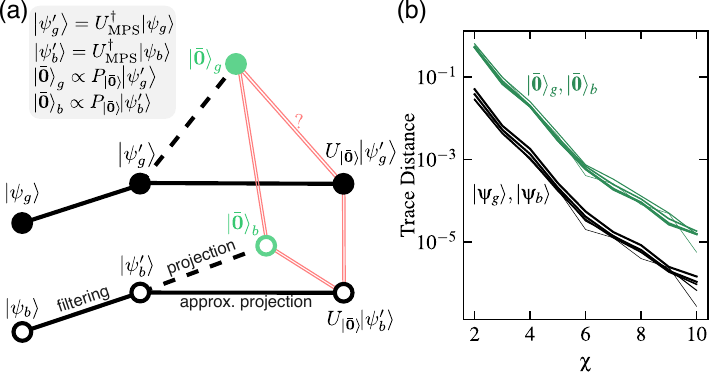}
    \caption{
        (a) Schematic relation among the quantum states appearing in the sequential projected-sampler construction. 
        Filled circles denote ideal target states that are generally not classically accessible, while open circles denote approximate states that are classically accessible in practice. 
        Dashed lines represent the non-unitary projection onto the subspace orthogonal to $\ket{\bm0}$. 
        (b) Trace distance as a function of the maximum bond dimension $\chi$ used in constructing $\ket{\psi_b(\chi)}$, for several relevant state pairs and system sizes $n=12,16,20,50,100$. 
        Here, $\hat U_{\rm MPS}$ is chosen to be the two-layer circuit filter obtained in Fig.~\ref{fig:mps-trajectory}.
    }
    \label{fig:trace-distance}
\end{figure}

We next examine the sequential projected-sampler construction based on the unitary approximation $\hat U_{\ket{\bar{\bm0}}}$ to the non-unitary projector $\hat P_{\ket{\bar{\bm0}}}$. 
In practice, instead of the exact ground state $\ket{\psi_g}$, we start from the approximate MPS state $\ket{\psi_b(\chi)}$ and define the filtered state 
\begin{equation}
\ket{\psi_b'}\coloneqq \hat U_{\rm MPS}^{\dagger}\ket{\psi_b(\chi)}, 
\end{equation}
where $\hat U_{\rm MPS}$ is obtained from the MPS-based circuit encoding discussed above. 
The corresponding normalized projected state is
\begin{equation}\label{eq:0_bar_b_mps}
\ket{\bar{\bm0}}_b
\coloneqq
\frac{\hat P_{\ket{\bar{\bm0}}}\ket{\psi_b'}}
{\|\hat P_{\ket{\bar{\bm0}}}\ket{\psi_b'}\|_2}.
\end{equation}
Our goal is to approximate the action of $\hat P_{\ket{\bar{\bm0}}}$ on $\ket{\psi_b'}$ by a unitary $\hat U_{\ket{\bar{\bm0}}}$ such that
\begin{equation} \label{eq:0b}
\ket{\bar{\bm0}}_b \approx \hat U_{\ket{\bar{\bm0}}}\ket{\psi_b'},
\end{equation}
using the MPS-based circuit-encoding method. 
The same unitary is then applied to the ideal filtered state $\ket{\psi_g'}$ in order to assess the resulting algorithmic error with respect to the ideal projected state $\ket{\bar{\bm0}}_g$.


To quantify this error, we use the trace distance
\begin{equation}
    \mathcal D(\ket{\psi_1},\ket{\psi_2})
    \coloneqq
    \sqrt{1-\abs{\braket{\psi_1}{\psi_2}}^2},
\end{equation}
which measures the distance between two pure states~\cite{nielsen2010quantum}. 
The relevant algorithmic error in the sequential projected-sampler construction is therefore the distance between the encoded projected state and the ideal projected state,
\begin{equation}
    \mathcal D\!\left(\hat U_{\ket{\bar{\bm0}}}\ket{\psi_g'},\ket{\bar{\bm0}}_g\right).
\end{equation}

As illustrated schematically in Fig.~\ref{fig:trace-distance}(a), the triangle inequality implies
\begin{equation}\label{eq:trace-distance}
\begin{aligned}
    \mathcal D\!\left(\hat U_{\ket{\bar{\bm0}}}\ket{\psi_g'},\ket{\bar{\bm0}}_g\right)
    \leqslant\;&
    \mathcal D\!\left(\hat U_{\ket{\bar{\bm0}}}\ket{\psi_g'},\hat U_{\ket{\bar{\bm0}}}\ket{\psi_b'}\right)\\
    &+\mathcal D\!\left(\hat U_{\ket{\bar{\bm0}}}\ket{\psi_b'},\ket{\bar{\bm0}}_b\right)\\
    &+\mathcal D\!\left(\ket{\bar{\bm0}}_g,\ket{\bar{\bm0}}_b\right).
\end{aligned}
\end{equation}
Equation~\eqref{eq:trace-distance} follows from applying the triangle inequality twice, with the intermediate states $\hat U_{\ket{\bar{\bm0}}}\ket{\psi_b'}$ and $\ket{\bar{\bm0}}_b$ inserted between $\hat U_{\ket{\bar{\bm0}}}\ket{\psi_g'}$ and $\ket{\bar{\bm0}}_g$. 
Using the invariance of the trace distance under unitary transformations~\cite{Li2017Efficient}, the first term becomes
\begin{equation}
    \mathcal D\!\left(\hat U_{\ket{\bar{\bm0}}}\ket{\psi_g'},\hat U_{\ket{\bar{\bm0}}}\ket{\psi_b'}\right)
    =
    \mathcal D\!\left(\ket{\psi_g},\ket{\psi_b(\chi)}\right).
\end{equation}
Thus, Eq.~\eqref{eq:trace-distance} separates the total algorithmic error into contributions from the finite-bond-dimension error of the MPS approximation, the circuit-encoding error of the projector approximation, and the mismatch between the ideal and approximate projected target states. 

Figure~\ref{fig:trace-distance}(b) shows that the distance $\mathcal D(\ket{\bar{\bm0}}_g,\ket{\bar{\bm0}}_b)$ decreases systematically with increasing bond dimension $\chi$. 
This indicates that the third contribution in Eq.~\eqref{eq:trace-distance} can be effectively controlled by improving the MPS approximation to the ground state. 
Overall, the sequential projected-sampler error is governed jointly by the accuracy of $\ket{\psi_b(\chi)}$ and the circuit-encoding quality of $\hat U_{\ket{\bar{\bm0}}}$. 
Since all states indicated by open circles in Fig.~\ref{fig:trace-distance}(a) are classically simulable, the sequential projected-sampler construction provides a practical route to preparing approximate projected samplers.


\begin{figure}[tbh]
    \includegraphics[width=\linewidth]{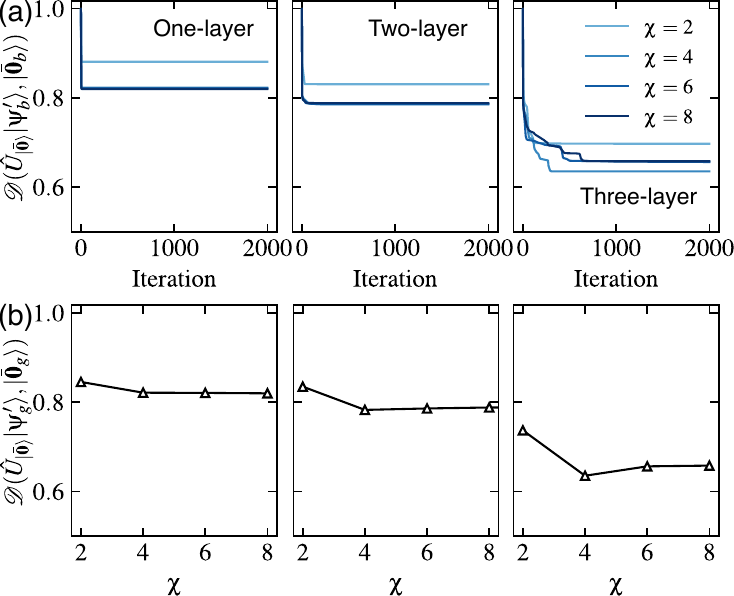}
    \caption{
        (a) Optimization trajectories of the MPS-based circuit-encoding algorithm for approximating the projector at system size $n=20$, with maximum bond dimension $\chi=2,4,6,8$ used in constructing $\ket{\psi_b(\chi)}$. 
        One iteration corresponds to a single forward sweep followed by a single backward sweep.
        (b) Trace distance between the encoded state $\hat U_{\ket{\bar{\bm0}}}\ket{\psi_g'}$ and the ideal projected state $\ket{\bar{\bm0}}_g$. 
        Here, $\hat U_{\ket{\bar{\bm0}}}$ is obtained from the optimization in (a) for $\chi=6$, and the same two-layer filter used to construct $\ket{\psi_b'}$ in (a) is also used in defining $\ket{\psi_g'}$. 
        From left to right, the encoded circuits for $\hat U_{\ket{\bar{\bm0}}}$ use one, two, and three brick-wall layers of local two-qubit unitaries.
    }
    \label{fig:projector-trajectory}
\end{figure}

We now examine the quality of the circuit encoding of the projector $\hat P_{\ket{\bar{\bm0}}}$ according to Eq.~\eqref{eq:0b}. 
Here, the encoded unitary $\hat U_{\ket{\bar{\bm0}}}$ is obtained by the MPS-based circuit-encoding method, and the filtered input state $\ket{\psi_b'}$ is generated by applying the two-layer filter $\hat U_{\rm MPS,II}^{\dagger}$ to $\ket{\psi_b(\chi)}$. 
Figure~\ref{fig:projector-trajectory}(a) shows the optimization trajectories of the trace distance 
\[
\mathcal D\!\left(\hat U_{\ket{\bar{\bm0}}}\ket{\psi_b'},\ket{\bar{\bm0}}_b\right)
\]
in the MPS-based circuit-encoding algorithm for $\chi=2,4,6,8$ and for encoded circuits $\hat U_{\ket{\bar{\bm0}}}$ composed of one, two, and three brick-wall layers of local two-qubit unitaries. 
All cases converge within 2000 iteration steps, and increasing the number of circuit layers generally reduces the trace distance. 
Figure~\ref{fig:projector-trajectory}(b) directly evaluates the resulting error
\[
\mathcal D\!\left(\hat U_{\ket{\bar{\bm0}}}\ket{\psi_g'},\ket{\bar{\bm0}}_g\right)
\]
for $\chi=6$ by applying the converged unitary $\hat U_{\ket{\bar{\bm0}}}$ to the ideal filtered state $\ket{\psi_g'}$. 
The resulting values closely track the converged trace distances shown in Fig.~\ref{fig:projector-trajectory}(a). 
We therefore conclude that, in the sequential projected-sampler construction, the dominant algorithmic error is determined primarily by the encoding quality of $\hat U_{\ket{\bar{\bm0}}}$ once the MPS approximation is sufficiently accurate.

\begin{figure}[tbh]
    \includegraphics[width=\linewidth]{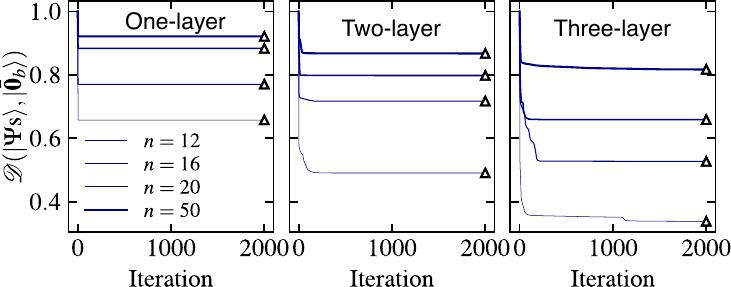}
    \caption{
        Optimization trajectories of the MPS-based circuit-encoding algorithm for direct encoding of the projected sampler $\ket{\Psi_{\rm s}}$ for system sizes $n=12,16,20,50$. 
        From left to right, the encoded circuits use one, two, and three brick-wall layers of local two-qubit unitaries. 
        One iteration corresponds to a single forward sweep followed by a single backward sweep. 
        Triangle markers indicate the trace distance between the converged $\ket{\Psi_{\rm s}}$ and the ideal projected state $\ket{\bar{\bm0}}_g$. 
        The bond dimension of $\ket{\psi_b(\chi)}$ is fixed to $\chi=6$.
            }
    \label{fig:sampler-trajectory}
\end{figure}

Finally, we consider the alternative strategy of directly encoding the projected sampler itself. 
Specifically, we directly encode $\ket{\bar{\bm0}}_b$ into a quantum circuit to obtain the sampler $\ket{\Psi_{\rm s}}$. 
In this analysis, we fix $\chi=6$ and use the two-layer circuit filter $\hat U_{\rm MPS,II}^{\dagger}$. 
Figure~\ref{fig:sampler-trajectory} shows the optimization trajectories of the trace distance for system sizes $n=12,16,20,50$ and for circuits with one, two, and three brick-wall layers of local two-qubit unitaries. 
After convergence, we evaluate the trace distance between the encoded sampler $\ket{\Psi_{\rm s}}$ and the ideal projected state $\ket{\bar{\bm0}}_g$. 
As in Fig.~\ref{fig:projector-trajectory}, the trace distances with respect to $\ket{\bar{\bm0}}_b$ and $\ket{\bar{\bm0}}_g$ are very similar, and increasing the number of circuit layers consistently improves the encoding quality. 
For $n=20$, the converged trace distances are comparable to those in Fig.~\ref{fig:projector-trajectory}(b) when the same number of encoding layers is used. 
For larger system sizes, however, the performance of direct encoding deteriorates, highlighting the need for improved circuit ans\"atze for efficient sampler preparation in larger systems.


\section{Experimental details}\label{appx:IV}

In this appendix, we summarize additional details of the quantum-hardware experiments performed on IBM Quantum's Heron R2 processor \texttt{ibm\_kobe}. 
Figure~\ref{fig:sqd-depth}(a) shows the qubit layouts and coupling paths used in the $n=20$, $50$, and $100$ implementations. 
The colored qubits correspond to those used for preparing the sampler $\ket{\Psi_{\rm s}}$. 
The colors on the coupling links indicate the CZ-gate error, while the colors of the qubits indicate the readout error. 
The relevant single- and two-qubit device properties at the time of the experiments are summarized in Table~\ref{tab:hardware}.


\begin{figure}[tbh]
    \includegraphics[width=\linewidth]{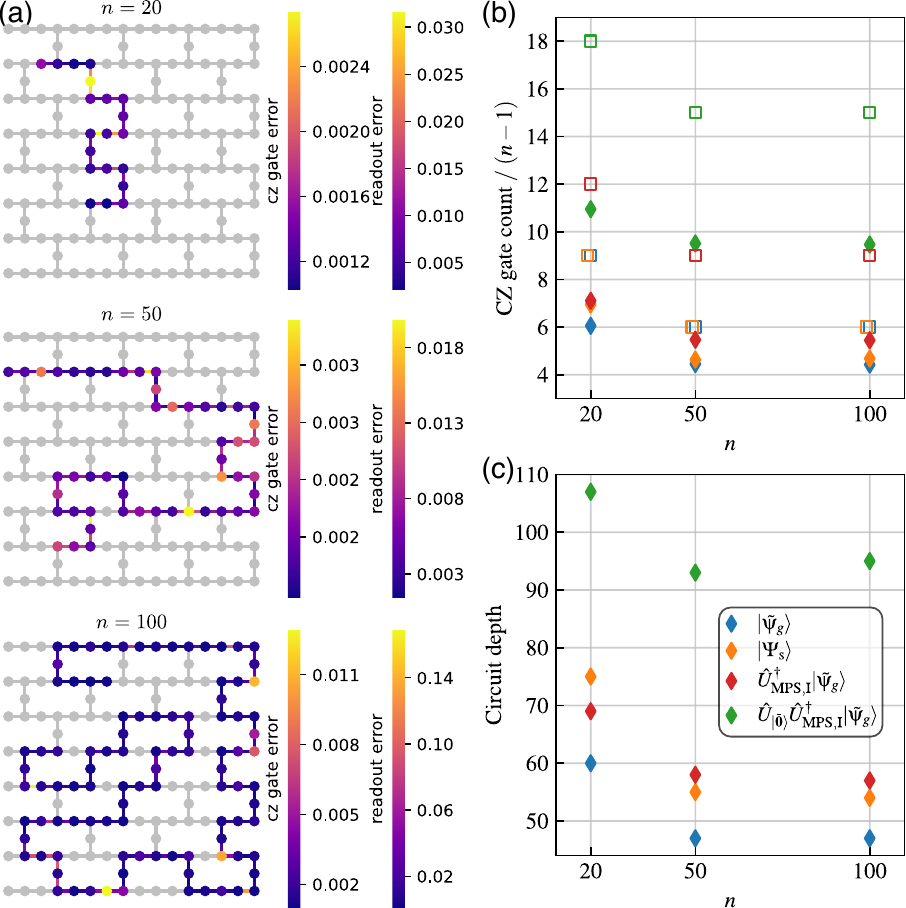}
    \caption{
      (a) Qubit layouts and coupling paths for the $20$-, $50$-, and $100$-qubit implementations on IBM Quantum's \texttt{ibm\_kobe}. 
      The colors on the links indicate the CZ-gate error, while the colors of the qubits indicate the readout error.  
      (b) Two-qubit CZ gate count of the transpiled quantum circuits (solid diamonds) for the samplers $\ket*{\tilde\psi_g}$, $\ket{\Psi_{\rm s}}$, $\umpsdI\ket*{\tilde\psi_g}$, and $\hat U_{\ket{\bar{\bm0}}}\umpsdI\ket*{\tilde\psi_g}$, plotted as functions of system size. 
      For comparison, the corresponding count for the naive circuit implementation of these sampler states is also shown by open squares. 
      Here, $\hat U_{\rm MPS,I}^{\dagger}$ denotes the one-layer filter. 
      The circuits preparing $\ket*{\tilde\psi_g}$ and $\ket{\Psi_{\rm s}}$ use three brick-wall layers of local two-qubit gates for $n=20$ and two such layers for $n=50$ and $100$, whereas $\hat U_{\ket{\bar{\bm0}}}$ is implemented using two brick-wall layers. 
      (c) Same as (b), but for the transpiled circuit depth, which includes single-qubit layers as well as two-qubits layers. 
      }
    \label{fig:sqd-depth}
\end{figure}

In the hardware implementation, we prepare the sampler states $\ket*{\tilde\psi_g}$, $\hat U_{\rm MPS,I}^{\dagger}\ket*{\tilde\psi_g}$, $\hat U_{\ket{\bar{\bm0}}}\hat U_{\rm MPS,I}^{\dagger}\ket*{\tilde\psi_g}$, and $\ket{\Psi_{\rm s}}$ on the quantum computer. 
A general two-qubit unitary in $SU(4)$ can be synthesized using an elementary gate set consisting of single-qubit rotations and three CNOT gates. 
In the present study of the quantum Ising model, however, all local two-qubit unitaries appearing in the encoded circuits can be chosen to be real, so that they belong to $O(4)$ in general. 
Moreover, we observe that approximately half of the local two-qubit unitaries obtained in our encoding procedure belong to $SO(4)$ and can therefore be decomposed using only two CNOT gates.
Accordingly, in a naive implementation of these sampler circuits, the total two-qubit gate count is bounded above by $3\times(n-1)\times$ (the total number of brick-wall layers of local two-qubit unitaries), as indicated by the open squares in Fig.~\ref{fig:sqd-depth}(b). 
%
Since the circuits are arranged in a brick-wall structure, the corresponding two-qubit gate depth is proportional to the number of encoding layers and thus remains essentially independent of system size once the number of layers is fixed. 
What degrades with increasing system size is not the nominal circuit depth itself, but rather the accuracy of the circuit encoding obtained with a fixed number of layers, as discussed in Appendix~\ref{appx:III}.

To reduce the two-qubit gate count further, we use \textit{Qiskit}'s transpiler with \texttt{optimization\_level=3}~\cite{Qiskit} to compile each circuit into 10 candidates compatible with the instruction-set architecture (ISA) of \texttt{ibm\_kobe}, and select the candidate with the smallest two-qubit gate count and circuit depth. 
As seen in Figs.~\ref{fig:sqd-depth}(b) and \ref{fig:sqd-depth}(c), the  sequentially encoded projected filtered sampler $\hat U_{\ket{\bar{\bm 0}}}\hat U_{\rm MPS,I}^{\dagger}\ket*{\tilde\psi_g}$ has the largest two-qubit gate count and circuit depth among the four samplers, whereas $\ket*{\tilde\psi_g}$ and $\ket{\Psi_{\rm s}}$ remain comparatively shallow over the system sizes considered.

\begin{table}[htb]
\caption{\label{tab:hardware}Single- and two-qubit properties of \texttt{ibm\_kobe} used in this study, including the qubit relaxation time $T_1$, qubit dephasing time $T_2$, readout error, and CZ error. }
\begin{ruledtabular}
\begin{tabular}{lllll}
{Property} & {Min.} & {Max.} & {Median} & {Avg.} \\
\hline
\multicolumn{5}{c}{$n=20$} \\
\hline
$T_1$ [$\mu$s] & $29.62$ & $360.2$ & $240.4$ & $216.2\pm81.82$ \\
$T_2$ [$\mu$s] & $20.05$ & $315.8$ & $147.0$ & $142.4\pm90.63$ \\
Readout [$\%$] & $0.220$ & $3.15$ & $0.537$ & $0.676\pm0.610$ \\
CZ [$\%$] & $0.103$ & $0.273$ & $0.150$ & $0.165\pm0.048$ \\
\hline
\multicolumn{5}{c}{$n=50$} \\
\hline
$T_1$ [$\mu$s] & $75.16$ & $402.5$ & $278.6$ & $272.4\pm73.87$ \\
$T_2$ [$\mu$s] & $19.53$ & $503.2$ & $193.7$ & $194.7\pm133.9$ \\
Readout [$\%$] & $0.146$ & $1.978$ & $0.464$ & $0.619\pm0.379$ \\
CZ [$\%$] & $0.097$ & $0.338$ & $0.154$ & $0.161\pm0.051$ \\
\hline
\multicolumn{5}{c}{$n=100$} \\
\hline
$T_1$ [$\mu$s] & $2.967$ & $410.0$ & $288.4$ & $276.3\pm73.80$ \\
$T_2$ [$\mu$s] & $4.506$ & $513.0$ & $179.5$ & $177.7\pm127.5$ \\
Readout [$\%$] & $0.146$ & $16.84$ & $0.610$ & $1.367\pm2.676$ \\
CZ [$\%$] & $0.103$ & $1.289$ & $0.164$ & $0.255\pm0.205$ \\
\end{tabular}
\end{ruledtabular}
\end{table}

\bibliography{FSQD.bib}

\end{document}